\documentclass[journal,12pt,onecolumn,draftcls]{IEEEtran}

\usepackage{amsthm}
\usepackage{cite}
\usepackage{amsmath}
\usepackage{amssymb}
\usepackage{graphicx}
\usepackage{graphics}
\usepackage{dsfont}
\usepackage{algorithm}
\usepackage{algorithmic}
\usepackage{epstopdf}
\usepackage{verbatim} 
\usepackage{setspace}

\makeatletter

\def\ps@IEEEtitlepagestyle{%
  \def\@oddfoot{\conf}%
  \def\@evenfoot{}%
}
\def\conf{%
  {\footnotesize Part of this work has been accepted for presentation in IEEE ICC 2015, London, UK, June 8-12, 2015.\hfill}
  \gdef\conf{}
}

\newsavebox{\theorembox}
\newsavebox{\lemmabox}
\newsavebox{\corollarybox}
\newsavebox{\propositionbox}
\newsavebox{\examplebox}
\newsavebox{\conjecturebox}
\newsavebox{\algbox}
\newsavebox{\qbox}
\newsavebox{\problembox}
\newsavebox{\definitionbox}
\newsavebox{\assumptionbox}
\newsavebox{\hypothesisbox}
\newsavebox{\factbox}
\savebox{\theorembox}{\noindent\bf Theorem}
\savebox{\lemmabox}{\noindent\bf Lemma}
\savebox{\corollarybox}{\noindent\bf Corollary}
\savebox{\propositionbox}{\noindent\bf Proposition}
\savebox{\examplebox}{\noindent\bf Example}
\savebox{\conjecturebox}{\noindent\bf Conjecture}
\savebox{\algbox}{\noindent\bf Algorithm}
\savebox{\qbox}{\noindent\bf Question}
\savebox{\definitionbox}{\noindent\bf Definition}
\savebox{\problembox}{\noindent\bf Problem}
\savebox{\assumptionbox}{\noindent\bf Assumption}
\savebox{\hypothesisbox}{\noindent\bf Hypothesis}
\savebox{\factbox}{\noindent\bf Fact}
\newtheorem{theorem}{\usebox{\theorembox}}
\newtheorem{lemma}{\usebox{\lemmabox}}
\newtheorem{corollary}{\usebox{\corollarybox}}

\newtheorem{definition}{\usebox{\definitionbox}}

\begin{document}
\title{Energy-Efficient Broadcasting for Cross Wireless Ad-Hoc Networks}

\author{
\IEEEauthorblockN{Mohammad R. Ataei, Amir H. Banihashemi and Thomas Kunz}

\IEEEauthorblockA{Systems and Computer Engineering Department, Carleton University, Ottawa, ON, Canada
\\
mrataei, ahashemi, tkunz@sce.carleton.ca}
}
\maketitle

\begin{abstract}
\boldmath
In this paper, we propose solutions for the energy-efficient broadcasting over cross networks, where $N$ nodes are located on two perpendicular lines. Our solutions consist of an algorithm which finds the optimal range assignment in polynomial time ($\mathcal{O}(N^{12})$), a near-optimal algorithm with linear complexity ($\mathcal{O}(N)$), and a distributed algorithm with complexity $\mathcal{O}(1)$. To the best of our knowledge, this is the first study presenting an optimal solution for the minimum-energy broadcasting problem for a 2-D network (with cross configuration). We compare our algorithms with the broadcast incremental power (BIP) algorithm, one of the most commonly used methods for solving this problem with complexity $\mathcal{O}(N^2)$. We demonstrate that our near-optimal algorithm outperforms BIP, and that the distributed algorithm performs close to it.
Moreover, the proposed distributed algorithm can be used for more general two-dimensional networks, where the nodes are located on a grid consisting of perpendicular line-segments.
The performance of the proposed near-optimal and distributed algorithms tend to be closer to the optimal solution for larger networks.
\end{abstract}


\section{Introduction}
\label{sec:int}
\IEEEPARstart{W}{ireless} ad-hoc networks have attracted more interest in recent years due to their numerous applications \cite{ad_hoc_book}. In these networks, broadcasting mechanisms are used for data exchange purposes, e.g., disseminating important messages or attaining path discovery information in routing algorithms \cite{routing}. The \emph{Minimum-Energy Broadcasting} problem in wireless networks focuses on finding a transmission range assignment for all the nodes in the network such that the total consumed energy for broadcasting data from one specific node, the source node, to all the other nodes is minimized \cite{MinEnBroad}. We consider the case where there are $N$ nodes in the network, and the exact location of them is known.

An optimal solution with complexity $\mathcal{O}(N^2)$ was presented in \cite{ataei} for the minimum-energy broadcasting problem in 1-D networks (linear networks). The minimum-energy broadcasting problem is known to be non-deterministically polynomial-time (NP) hard for D-dimensional spaces with $D \geq 2$, \cite{np_hard,Clementi_2_3_D,Kranakis_strong}. One main reason for the difficulty of this problem is the \emph{wireless multicast advantage} \cite{BIP,BIP_Journal}, i.e., the reception of the transmitted data by multiple nodes within the range of a single transmission.
There have been a number of heuristic approaches for solving this problem. For a survey on the existing works on minimum-energy broadcasting (as a special case of multicasting), one can see \cite{Multicast_survey}.

The broadcast incremental power (BIP) algorithm, with complexity $\mathcal{O}(N^3)$ \cite{BIP,BIP_Journal}, is one of the most commonly used methods for the energy-efficient broadcasting problem in 2-D networks. In \cite{2D_Optimal_Bound}, it is shown that the approximation ratio\footnote{The ratio of the energy consumption of a given assignment to that of the optimal assignment is called the approximation ratio of that assignment.} of BIP, unlike some other well-known algorithms, e.g., shortest-path tree (SPT) \cite{SPT_ref}, is a constant (does not depend on the total number of the nodes in the network). This constant value is lower than the approximation ratio of the other algorithms also studied in \cite{2D_Optimal_Bound}. In BIP, the goal is to construct a Minimum Spanning Tree (MST) of the network graph\footnote{The MST of a graph is the spanning tree with weight less than or equal to the weight of every other spanning tree of that graph.} starting from the source, and adding new nodes one by one to the tree. The cost function, which is to be minimized in the construction, is the incremental cost of adding each additional node. This incremental cost is defined as the minimum additional power required of some node in the current tree to reach the new node. Recently in \cite{BIP_N2_Complexity}, the complexity of BIP is reduced to $\mathcal{O}(N^2)$ and the lower bound on the approximation ratio of BIP is improved.

A modification to BIP, called \emph{sweep}, is proposed in \cite{BIP,BIP_Journal}, to improve the power consumption. In this procedure, unnecessary transmissions are eliminated. It is shown in \cite{sweep_complexity} that the complexity of the sweep procedure is also $\mathcal{O}(N^2)$.

Another frequently cited method for energy-efficient broadcasting is the algorithm based on finding the Euclidean MST of the graph representing the network rooted at the source node \cite{Nguyen}. The network graph is constructed by considering the nodes of the network as vertices of the graph. There is an edge between any two nodes, with weight equal to the Euclidean distance between the nodes.
In this algorithm, the MST of the network graph is found (e.g., using Prim's algorithm
). In the MST-based range assignment, for each node, the maximum weight of the edge between that node and its children\footnote{For any node in a tree, the nodes directly below it in the tree hierarchy are called the children of that node.} is assigned as the node's transmission range, i.e.,

\begin{align}
R^{MST}(i)=\max_{u:(i,u)\in MST}\{d(i,u)\},
\end{align}
where $R^{MST}(i)$ denotes the assigned range to node $i$, and the notation $(i,u)$ is used to denote the edge between nodes $i$ and $u$ in the MST, where $i$ is the parent of $u$. The Euclidean distance between nodes $a$ and $b$ is denoted by $d(a,b)$.

The MST-based algorithm has complexity $\mathcal{O}(N^2)$ \cite{Nguyen}. In \cite{grid,grid_clementi}, it is shown that even for 2-D networks with special topologies, i.e., when the nodes are located at the intersection points of a square grid, the MST-based range assignment is far from optimal. It is worth mentioning that the MST-based range assignment is the optimal solution for the minimum-energy broadcasting problem in wired networks. The performance of BIP is shown to be better than that of the MST algorithm for general 2-D wireless networks \cite{BIP,BIP_Journal}.

In this paper, we are interested in finding both optimal and low-complexity solutions for the broadcasting problem in a special 2-D wireless network, where the nodes are located on a cross, consisting of two perpendicular lines. This kind of network is used to model perpendicular roads for VANETs \cite{vanet}, or the sensor field of power grid monitoring systems (or other sensor networks) \cite{sensor,Power_Monitoring_Opportunities}.
We provide three range assignment solutions: optimal, near-optimal and distributed. To the best of our knowledge, our work is the first study presenting an optimal solution for the minimum-energy broadcasting problem for a 2-D network.

We show that the optimal solution for cross networks can be found in polynomial time, but has a rather high complexity ($\mathcal{O}(N^8)$ or $\mathcal{O}(N^{12})$ for the cases where the source node is located at the intersection of the lines, or somewhere other than the intersection, respectively). The proposed centralized near-optimal algorithm has complexity $\mathcal{O}(N)$, and outperforms BIP. The proposed distributed algorithm performs close to BIP (and BIP with sweep) and has a very low complexity of $\mathcal{O}(1)$.
Furthermore, we prove that our distributed algorithm has the same outcome as the MST algorithm for networks on a cross, but with much lower complexity ($\mathcal{O}(1)$ instead of $\mathcal{O}(N^{2})$). 
As the size of the network increases, the proposed near-optimal and distributed algorithms perform closer to optimal.

The rest of the paper is organized as follows. Section \ref{sec:sys} describes the system model. The proposed range assignments are explained in Section \ref{sec:pro}. A special case of the networks in which the source node is located at the intersection is studied in Section \ref{sec:atintersect}. The use of our proposed distributed algorithm in grid networks is briefly described in Section \ref{sec:grid}. We evaluate our proposed algorithms by comparing them to the existing works in Section \ref{sec:num}. The paper concludes with Section \ref{sec:con}.

\section{System Model}
\label{sec:sys}
We consider a set of $N$ nodes (denoted by $\mathcal{N}$) located on a cross, including a specific node, $s$, as the source. We assume that the exact location of the nodes is known. The source node can be located anywhere on the cross and broadcasts the data to all the other nodes in the network. This is to be performed in an energy-efficient multi-hop fashion.

We assume that all the nodes are equipped with an omnidirectional transmitter. Moreover, in our study, similar to \cite{BIP,BIP_Journal,sweep_complexity,Nguyen,grid,grid_clementi,np_hard,Clementi_2_3_D,Kranakis_strong,ataei}, we do not consider other issues such as channel contention and interference.

To solve this broadcasting problem, we need to assign a \emph{transmission range} to each node so that the total consumed energy is minimized, while the data is delivered to all the nodes of the network.
A \emph{transmission range assignment} $R$ is a function $R:\mathcal{N} \rightarrow \mathds{R}^+$, where $R(i)$ is the assigned transmission range to node $i\in \mathcal{N}$. We denote the consumed energy of the range assignment $R$ by $cost(R)$ and assume that it can be calculated, up to a constant multiplicative factor, using the following equation:
\begin{align}
cost(R)=\sum_{k\in \mathcal{N}}{R^\alpha(k)},
\end{align}
where $\alpha$ is the \emph{path-loss exponent} whose value is normally between 2 and 6 \cite{alpha}. By using the \emph{Minimum-Energy Range Assignment} (a.k.a, Optimal Range Assignment), denoted by $R^*$, every node in the network will receive the data transmitted by the source node with the minimum possible cost.

A node is assumed to have circular coverage up to its transmission range, and any other node located in the transmission range of this node can receive the transmitted data. Some definitions and notations are presented in Fig. \ref{fig1}.

\begin{figure}[h!]
  \centering
    \includegraphics[width=4in]{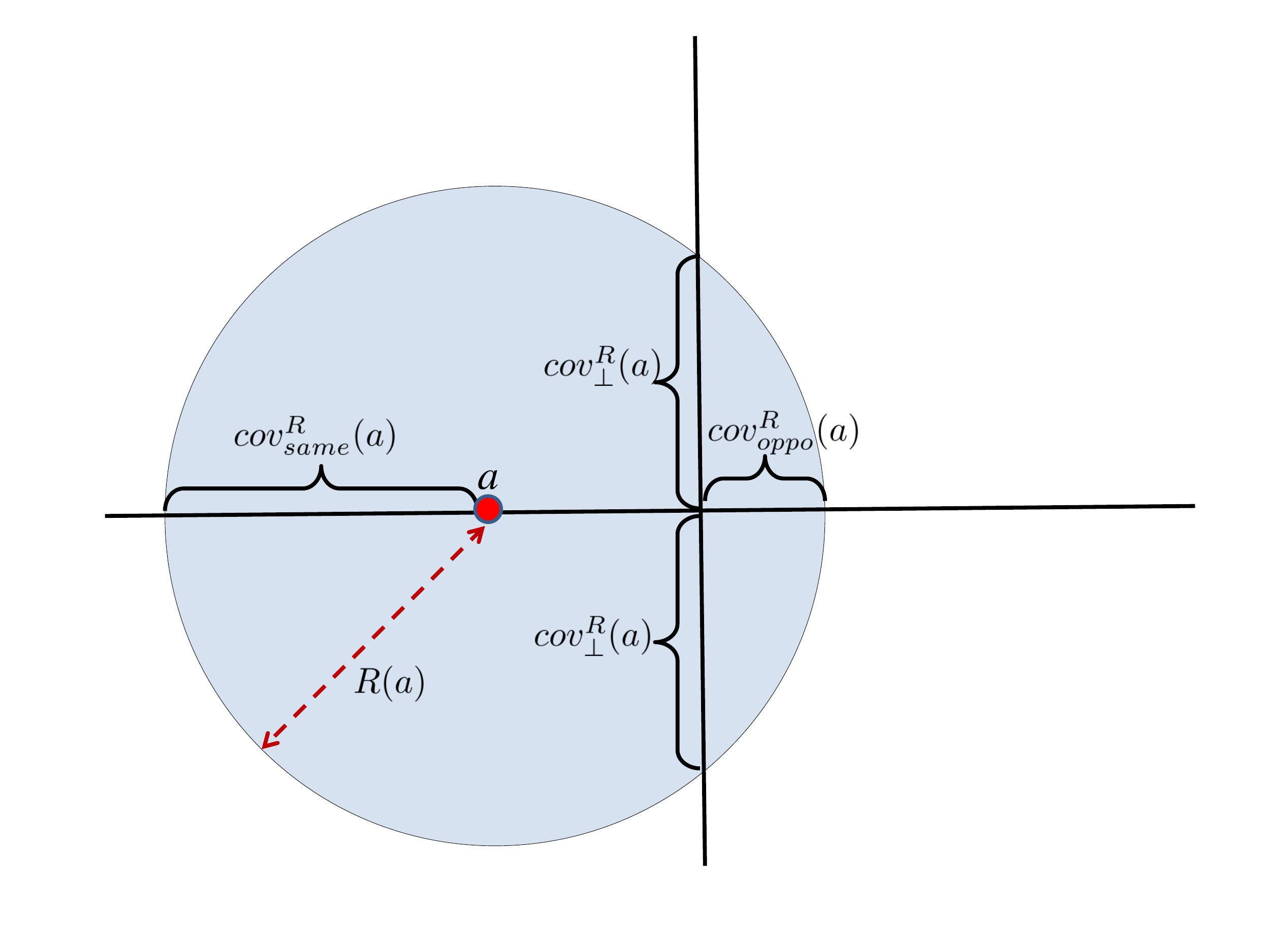}
     \caption{Some definitions/notations for networks on a cross.}
     \label{fig1}
\end{figure}

According to the circular transmission range of the nodes, as shown in Fig. 1, for every node $a\in\mathcal{N}$, we have $R(a)=cov^{R}_{same}(a)\geq cov^{R}_{\bot}(a)\geq cov^{R}_{oppo}(a)$.

In the following, we provide some definitions needed throughout the paper.

\begin{definition}
A cross network has five \textbf{segments}, as shown in Fig. \ref{fig4}. Segments I, III, IV and V are half-lines, while Segment II is the line-segment bounded by the source node and the intersection of the two perpendicular lines. For any node $a\in \mathcal{N}\setminus\{s\}$,\footnote{For any two sets $\mathcal{A}$ and $\mathcal{B}$, notation $\mathcal{A}\setminus \mathcal{B}$ denotes the set of all elements which are members of $\mathcal{A}$ but not members of $\mathcal{B}$.} we denote the segment on which it is located by $\mathcal{S}_a$.\footnote{By definition, the source node is not on any segment.}
\end{definition}

\begin{figure}[h!]
  \centering
    \includegraphics[width=4in]{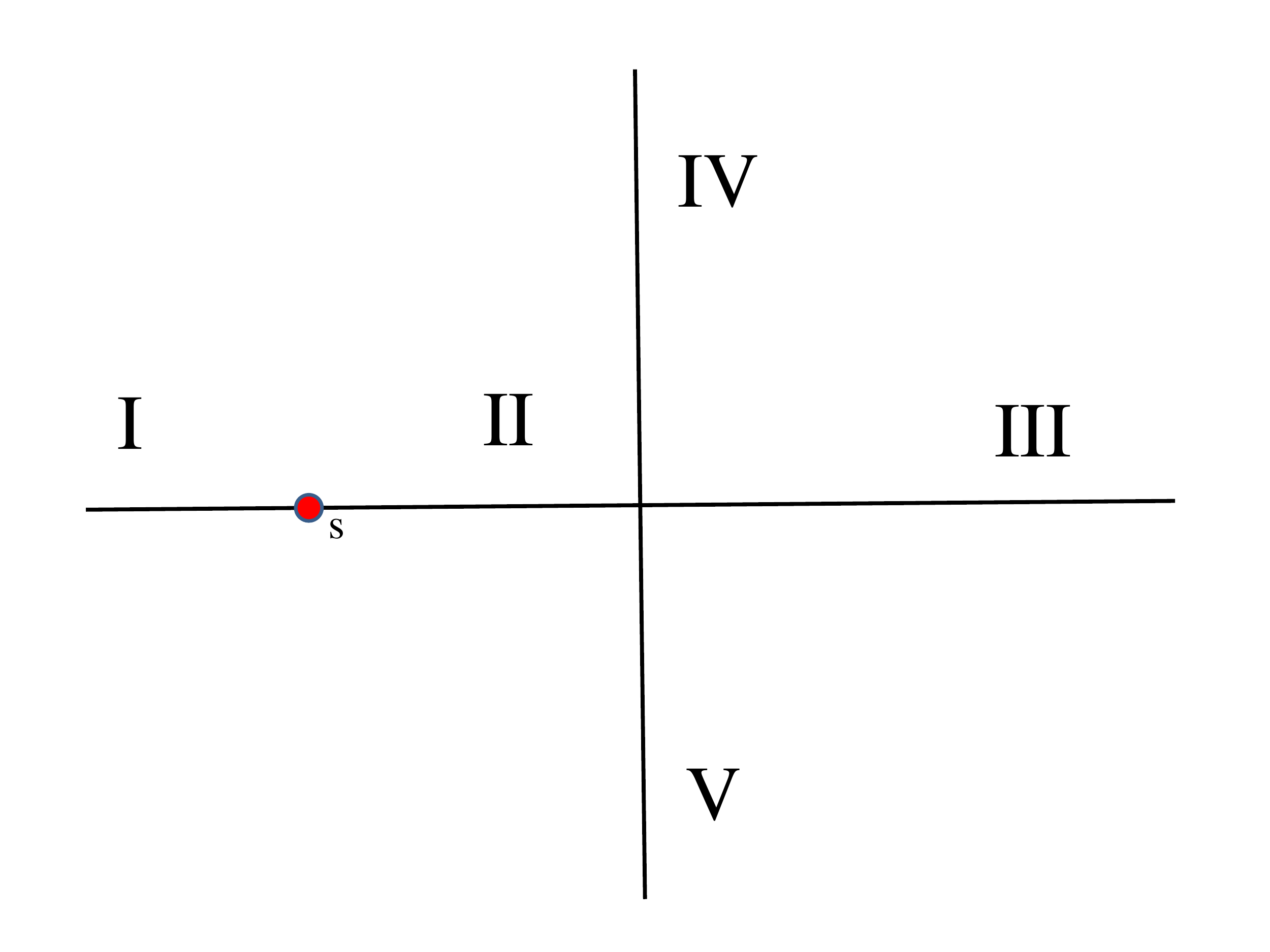}
     \caption{Segmentation of cross networks.}
     \label{fig4}
\end{figure}

\begin{definition}
On each segment, the closest node to the source node is called the first node of that segment.
The last node of a segment is the node farthest away from the source on that segment.
We denote the first node and the last node of any segment $\mathcal{S}$ by $f_{\mathcal{S}}$ and $l_{\mathcal{S}}$, respectively.\footnote{\label{footII}If Segment II is empty, we assume that node $s$ takes all the functionalities of node $l_{\text{II}}$.}
\end{definition}

We denote the set of nodes $\mathcal{N}\setminus\{s,l_{\text{II}},f_{\text{III}},f_{\text{IV}},f_{\text{V}}\}$ by $\mathcal{\hat{N}}$.

\begin{definition}
Node $a$ is after node $b$ on the same segment, if $\mathcal{S}_a=\mathcal{S}_b$ and $d(s,a)>d(s,b)$. For these two nodes, we say node $b$ is before node $a$.
\end{definition}

\begin{definition}
For any node $a\in \mathcal{N}\setminus\{s\}$, we call the first node after $a$ on $\mathcal{S}_a$, \textbf{the next adjacent neighbor} of node $a$, and we denote it by $\mathfrak{n}_a$. Furthermore, we define $M(a)=d(a,\mathfrak{n}_a)$. If node $a$ is the last node on its segment, as it does not have a next adjacent neighbor, we define $M(a)=0$.
\end{definition}

Fig. \ref{fignan} illustrates the next adjacent neighbor of some nodes in a cross network.

\begin{figure}[h!]
  \centering
    \includegraphics[scale=0.45, trim = 5mm 10mm 10mm 10mm, clip]{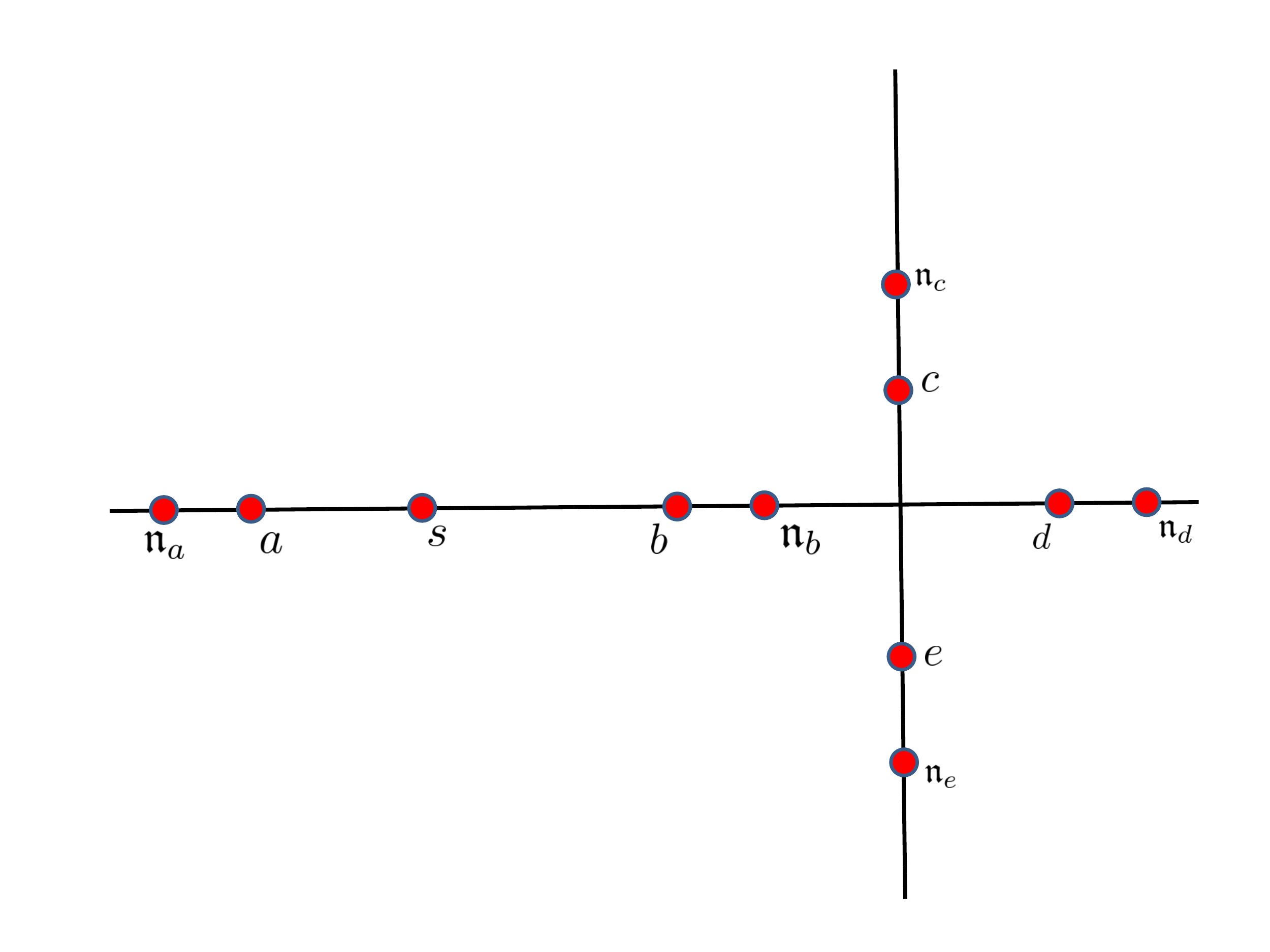}
     \caption{The next adjacent neighbors of nodes $a$, $b$, $c$, $d$ and $e$.}
     \label{fignan}
\end{figure}

\begin{definition}
For assignment $R$, we say that node $a\in \mathcal{N}\setminus\{s\}$ has \textbf{increased (transmission) range}, if $R(a)>M(a)$.
\end{definition}

\begin{definition}
In a cross network utilizing transmission range assignment $R$, for any node $a\in \mathcal{N}\setminus\{s\}$, we call the set of all nodes located after node $a$ on $\mathcal{S}_a$, that are within the transmission range of node $a$, \textbf{the same-segment receivers} of node $a$. We call the set of nodes on segments other than $\mathcal{S}_a$ that are within the transmission range of node $a$, \textbf{the other-segment receivers} of node $a$. The union of these two sets for node $a$ is called \textbf{the receivers} of node $a$.
\end{definition}

\begin{definition}
In a cross network utilizing transmission range assignment $R$, for any two nodes $a,b\in \mathcal{N}$, node $b$ receives the data via a path (starting from the source node) containing node $a$, if for node $b$ to receive the data, node $a$ has to transmit with $R(a)(\neq 0)$. In other words, if node $a$ does not transmit, node $b$ will not receive the data (at all). We show this relationship by $b\overset{R}{\leftarrow}a$. We show the case where node $b$ can receive the data even if node $a$ does not transmit by $b\overset{R}{\nleftarrow}a$.
\end{definition}

This concept is important to note, because a node has to receive the data first to be able to transmit it to other nodes. So, if $b\overset{R}{\leftarrow}a$, first node $a$ has to transmit the data with $R(a)$, then node $b$ will be able to transmit it to other nodes.

Also note that for nodes $a$ and $b$, the relation $b\overset{R}{\leftarrow}a$ does not necessarily mean that node $b$ receives the data in just one hop from node $a$. It means that the data travels through node $a$ to get to node $b$, and this is the only way for node $b$ to receive it.

Another important property of this concept is that for any two nodes $a,b\in \mathcal{N}$, we can have either $a\overset{R}{\leftarrow}b$ or $b\overset{R}{\leftarrow}a$, and not both. But we can have $a\overset{R}{\nleftarrow}b$ or $b\overset{R}{\nleftarrow}a$ or both.

\begin{definition}
In a cross network utilizing transmission range assignment $R$, we call those receivers $b$ of node $a\in \mathcal{N}$, \textbf{the intended receivers} of $a$, if for all of them, we have $b\overset{R}{\leftarrow}a$. We denote the set of intended receivers of node $a$ with transmission range $R(a)$ by $\mathfrak{I}^{R}_a$.
\end{definition}

\section{Proposed Range Assignments}
\label{sec:pro}

In the following, we explain our proposed range assignments, starting from the optimal one, followed by the near-optimal assignment and finally the distributed one.

\subsection{Optimal Range Assignment}
\label{optBsection}
One of our main contributions is to prove that in the optimal range assignment, there exists a small (and independent from $N$) set of nodes with increased transmission range, and the other nodes have either $0$ or the distance to their next adjacent neighbor as their transmission range. Also, we provide an algorithm with polynomial complexity in $N$ to find the optimal assignment for the nodes on a cross.

An upper bound of eight on the number of nodes with increased transmission range in the optimal range assignment is established in the following theorem.

\begin{theorem}
\label{theB1}
In the optimal range assignment for a cross network, denoted by $R^*$, there exist at most three nodes with increased transmission range in $\mathcal{\hat{N}}$. All the other nodes in $\mathcal{\hat{N}}$ have either $0$ or the distance to their next adjacent neighbor as their transmission range. The optimal transmission range of a node in the set $\{s,l_{\text{II}},f_{\text{III}},f_{\text{IV}},f_{\text{V}}\}$ can be equal to its distance to any other node in the network.
\end{theorem}

For proving this theorem, we first introduce some lemmas.

\begin{lemma}
\label{lem0}
For any node $a\in\mathcal{N}$, if $R^*(a)\neq 0$, then $\mathfrak{I}^{R^*}_a\neq \emptyset$.
\end{lemma}

\begin{lemma}
\label{lem1}
For any given set $\{a_1,a_2,\cdots,a_W\}$ of positive numbers, where $W$ is an arbitrary integer, and for any  $\alpha \geq 2$, we have:
\begin{align}
 \left ( \sum_{k=1}^{W}{a_k}\right )^\alpha \geq \sum_{k=1}^{W}{a_k^\alpha}.
\end{align}
\end{lemma}

\begin{corollary}
\label{cor1}
To transmit data on a line from node $a$ to node $b$ with minimum energy consumption, the data has to be transmitted hop by hop using the nodes in between transmitter and receiver. This means that node $a$ transmits the data to node $\mathfrak{n}_a$, and node $\mathfrak{n}_a$ transmits to its next adjacent neighbor, and so forth till the data reaches node $b$. We denote this transmission scheme by $a \dashrightarrow b$, where all the nodes from $a$ up to $b$ use their $M$ values as their transmission ranges.
\end{corollary}

\begin{lemma}
\label{lem2}
In a cross network utilizing transmission range assignment $R$, if node $a\in \mathcal{N}\setminus\{s\}$ receives the data, all the nodes before it on its segment have also received the data.
\end{lemma}

\begin{proof}
This is a direct result of the circular transmission range assumption of the nodes.
\end{proof}

\begin{corollary}
\label{cor2}
In a cross network utilizing transmission range assignment $R$, for nodes $a,b \in \mathcal{N}\setminus\{s\}$ where $\mathcal{S}_a=\mathcal{S}_b$, and node $b$ is after node $a$, we have $a\overset{R}{\nleftarrow}b$.
\end{corollary}

\begin{lemma}
\label{lem3}
In a cross network utilizing transmission range assignment $R$, for nodes $a\in \mathcal{N}\setminus\{s\}$ and $c\in \mathcal{N}$, if $a\overset{R}{\leftarrow}c$, then we have $b\overset{R}{\leftarrow}c$ for all nodes $b$ after node $a$ on $\mathcal{S}_a$.
\end{lemma}

\begin{proof}
The proof is by contradiction. Suppose for two nodes $a$ and $b$ on the same segment, where node $b$ is after node $a$ we have $a\overset{R}{\leftarrow}c$, but $b\overset{R}{\nleftarrow}c$. This means that if node $c$ does not transmit the data, then node $a$ will not receive it, but node $b$ receives it from another path. This contradicts Lemma \ref{lem2}.
\end{proof}

The contraposition of Lemma \ref{lem3} gives us the following corollary.

\begin{corollary}
\label{cor3}
In a cross network utilizing transmission range assignment $R$, if $b\overset{R}{\nleftarrow}c$ for node $b\in \mathcal{N}\setminus\{s\}$, then we have $a\overset{R}{\nleftarrow}c$ for all nodes $a$ before node $b$ on its segment.
\end{corollary}

Throughout the paper, we implicitly use the following lemma.

\begin{lemma}
\label{lem4}
In a cross network utilizing transmission range assignment $R$, any node $a\in \mathcal{N}$ is the intended receiver of at most one other node.
\end{lemma}

\begin{lemma}
\label{lem5}
In a cross network utilizing transmission range assignment $R$, if for nodes $a,b\in \mathcal{N}$ we have $b\overset{R}{\nleftarrow}a$, then the intended receivers of node $a$ on any segment $\mathcal{S}$ (if any exists) are after the last receiver of $b$ on $\mathcal{S}$, which is denoted by $r^{R}_{b,\mathcal{S}}$ (if this node exists).
\end{lemma}

\begin{proof}
Node $b$ receives the data via a path that does not contain node $a$, so if node $a$ does not transmit, node $b$ and (according to Lemma \ref{lem2}) all the nodes from $f_{\mathcal{S}_b}$ up to $b$ on segment $\mathcal{S}_b$ receive the data. Having the circular transmission range of the nodes in mind, we can see that when node $b$ transmits the data, all the nodes from $f_{\mathcal{S}}$ (if $\mathcal{S}\neq\mathcal{S}_b$) or node $b$ (if $\mathcal{S}=\mathcal{S}_b$) up to $r^{R}_{b,\mathcal{S}}$ on segment $\mathcal{S}$ receive the data. So all of the nodes from $f_{\mathcal{S}}$ to $r^{R}_{b,\mathcal{S}}$ on segment $\mathcal{S}$, receive the data even if node $a$ does not transmit. Hence they are not in the set $\mathfrak{I}^{R}_a$. Therefore, for node $a$ to have some intended receivers on segment $\mathcal{S}$, they have to be after node $r^{R}_{b,\mathcal{S}}$.
\end{proof}

\begin{lemma}
\label{lemB1}
For every node $a\in \mathcal{\hat{N}}$, $R^*(a)$ is either zero or greater than or equal to $M(a)$. Furthermore, if this node has an increased transmission range, i.e., if $R^*(a)>M(a)$, it must have at least one intended receiver on a segment other than $\mathcal{S}_a$.
\end{lemma}

\begin{proof}
According to Lemma \ref{lem0}, node $a$ does not have $0<R^*(a)<M(a)$, if it has no intended receivers. Using contradiction, suppose $R^*(a)<M(a)$, while it has some intended receivers. This implies that node $a$ does not have any receivers on $\mathcal{S}_a$, so all the nodes in $\mathfrak{I}^{R^*}_a$ are on other segments. We denote the farthest receiver of node $a$ when it transmits with $R^*(a)$ by $r^{R^*}_a$. For the node $k$ defined below, we have $d(a,r^{R^*}_a)>d(k,r^{R^*}_a)$.
\begin{align}
\label{kforproof}
k=\left\{
    \begin{array}{ll}
      s, & \hbox{if $a$ is on Segment I;} \\
      s, & \hbox{if $a$ is on Segment II};\footnotemark\\
      f_{\text{III}}, & \hbox{if $a$ is on Segment III;} \\
      f_{\text{IV}}, & \hbox{if $a$ is on Segment IV;} \\
      f_{\text{V}}, & \hbox{if $a$ is on Segment V.}
    \end{array}
  \right.
\end{align}
\footnotetext{If node $a$ is located on Segment II, it must have $R^*(a)>M(a)$ to cover some nodes on Segments III, IV or V. Hence, if $0<R^*(a)<M(a)$, all the other-segment receivers of $a$ are located on Segment I.}

Since node $a$ does not have any same-segment receivers, if we have another assignment $R$ with $R(a)=0$, $R(k)=\max\{R^*(k),d(k,r^{R^*}_a)\}$, and $R(i)=R^*(i)$, $\forall i\in \mathcal{N}\setminus\{k,a\}$, all the nodes will receive the data, but less energy will be consumed. This contradicts the optimality of $R^*$.

The proof of the next part is also by contradiction. Suppose $R^*(a)>M(a)$, but all of the intended receivers of node $a$ are on its own segment. According to Corollary \ref{cor1}, and since all the intended receivers of node $a$ are on the same line, by using $a \dashrightarrow r^{R^*}_{a,\mathcal{S}_a}$ the data will be delivered to all the intended receivers of node $a$ with less energy. In this case node $a$ does not have increased range, which contradicts our assumption.
\end{proof}

\begin{lemma}
\label{lemBII}
In a cross network utilizing transmission range assignment $R$, a node on one of the Segments III, IV or V does not have intended receivers on Segment II.
\end{lemma}

\begin{proof}
As a result of the circular transmission range assumption of the nodes, if a node on one of the Segments III, IV or V receives the data, all the nodes on Segment II have already received the data. 
\end{proof}

\begin{lemma}
\label{lemB4}
Consider a cross network that utilizes a range assignment $R$, for which $R(a)\neq 0$ and $R(b)\neq 0$ for nodes $a,b\in \mathcal{\hat{N}}$ (node $a$ being closer to source, if they are on the same segment). If node $a$ has an intended receiver after the last same-segment receiver of node $b$ (node $r^{R}_{b,\mathcal{S}_b}$), then assignment $R$ is not the optimal assignment.
\end{lemma}

\begin{proof}
The proof is by contradiction. Suppose $R=R^*$, and node $a\in\mathcal{\hat{N}}$ has an intended receiver after the last same-segment receiver of node $b\in\mathcal{\hat{N}}$. All the same-segment receivers of node $b$ are within the transmission range of node $a$. Three cases may exist:

1) Nodes $a$ and $b$ are on the same segment and $cov^{R}_{same}(a)\geq cov^{R}_{same}(b)$.

2) Node $a$ is on a segment aligned to the segment of node $b$.

3) Node $a$ is on a segment perpendicular to the segment of node $b$.

Case (1):
According to Corollary \ref{cor2}, $a\overset{R}{\nleftarrow}b$. The whole transmission circle of node $b$ is within the range of node $a$. So all the receivers of $b$ can receive the data from $a$, i.e., $\mathfrak{I}^{R}_b=\emptyset$, which contradicts the optimality of assignment $R$ (Lemma \ref{lem0}).

Case (2):
For node $a$ to have an intended receiver on $\mathcal{S}_b$, we can not have $\mathcal{S}_a=\text{III}$ and $\mathcal{S}_b=\text{II}$ (according to Lemma \ref{lemBII}). For the other cases, node $a$ having an intended receiver after $r^{R}_{b,\mathcal{S}_b}$, results in having the whole circular range of $b$ within the range of $a$. Hence, all the receivers of $b$ can receive the data from $a$. To be sure about the reception of data by node $a$, we define node $k$ as follows:

\begin{align}
k=\left\{
    \begin{array}{ll}
      s, & \hbox{if $b$ is on Segment I, and $a$ is on either Segments II or III;} \\
      f_{\text{III}}, & \hbox{if $b$ is on Segment III, and $a$ is on Segment II;} \\
      s, & \hbox{if $b$ is on Segment III, and $a$ is on Segment I;} \\
      f_{\text{IV}}, & \hbox{if $b$ is on Segment IV, and $a$ is on Segment V;} \\
      f_{\text{V}}, & \hbox{if $b$ is on Segment V, and $a$ is on Segment IV.}
    \end{array}
  \right.
\end{align}

We can have another assignment, $R'$, with $R'(b)=0$, $R'(i)=R(i)$, $\forall i\in \mathcal{N}\setminus\{k,b\}$, and $R'(k)=\max\{d(k,ro^{R}_b),R(k)\}$, where $ro^{R}_b$ denotes the farthest other-segment receiver of node $b$ when it transmits with $R(b)$. Note that since $k\overset{R}{\nleftarrow}b$ (Corollary \ref{cor2}), node $b$ has no effect on the delivery of data to node $k$. Using $R'$ all the nodes will receive the data with less energy, which contradicts the optimality of assignment $R$.

Case (3):
Similar to the previous case, and according to Lemma \ref{lemBII}, for node $a$ to have an intended receiver on $\mathcal{S}_b$, we can not have $\mathcal{S}_a=\text{IV}$ or $\mathcal{S}_a=\text{V}$, while $\mathcal{S}_b=\text{II}$. Define node $k'$ as follows:
\begin{align}
k'=\left\{
    \begin{array}{ll}
      s, & \hbox{if $b$ is on Segment I, and $a$ is on either Segments IV or V;} \\
      f_{\text{III}}, & \hbox{if $b$ is on Segment III, and $a$ is on either Segments IV or V;} \\
      f_{\text{IV}}, & \hbox{if $b$ is on Segment IV, and $a$ is on one of the Segments I, II or III;} \\
      f_{\text{V}}, & \hbox{if $b$ is on Segment V, and $a$ is on one of the Segments I, II or III.}
    \end{array}
  \right.
\end{align}

We can have another assignment, $R'$, with $R'(b)=0$, $R'(i)=R(i)$, $\forall i\in \mathcal{N}\setminus\{k',b\}$, and $R'(k')=\max\{d(k',ro^{R}_b),R(k')\}$, where $ro^{R}_b$ denotes the farthest other-segment receiver of node $b$ when it transmits with $R(b)$. Note that since $k'\overset{R}{\nleftarrow}b$ (Corollary \ref{cor2}), node $b$ has no effect on the delivery of data to node $k'$. Using $R'$ all the nodes will receive the data with less energy, which contradicts the optimality of assignment $R$.
\end{proof}

The following corollary is a direct result of Lemma \ref{lemB4}.

\begin{corollary}
\label{corBnan}
If node $a\in\mathcal{\hat{N}}$ does not have increased transmission range, and node $\mathfrak{n}_a$ receives the data from another node, then $R^*(a)=0$.
\end{corollary}

\begin{lemma}
\label{lemB2incd}
Consider a cross network that utilizes the optimal assignment $R^*$. For nodes $a,b\in \mathcal{\hat{N}}$, if $R^*(a)>M(a)$ and $R^*(b)>M(b)$, we have either $a\overset{R^*}{\leftarrow}b$ or $b\overset{R^*}{\leftarrow}a$.
\end{lemma}

\begin{proof}
The proof is by contradiction. Suppose nodes $a$ and $b$ have increased range, and we have $a\overset{R^*}{\nleftarrow}b$ and $b\overset{R^*}{\nleftarrow}a$. Since nodes $a$ and $b$ have increased range, according to Lemma \ref{lemB1}, they must have intended receivers on segments other than their own. According to Lemmas \ref{lem5} and \ref{lemB4}, none of them has other-segment intended receivers on $\mathcal{S}_a$ and $\mathcal{S}_b$. For any of the other segments, e.g., $\mathcal{S}$, using Lemma \ref{lem5}, the intended receivers of node $a$ on $\mathcal{S}$ have to be after $r^{R^*}_{b,\mathcal{S}}$ and the intended receivers of node $b$ on $\mathcal{S}$ have to be after $r^{R^*}_{a,\mathcal{S}}$, which is impossible and contradicts our assumption.
\end{proof}

\begin{lemma}
\label{lemBonII}
In a cross network, utilizing $R^*$, a node in $\mathcal{\hat{N}}$ with increased transmission range on Segment II must have at least one intended receiver on Segment I.
\end{lemma}

\begin{proof}
The proof is by contradiction. Suppose node $a$ on Segment II, has increased range in $R^*$, but does not have any intended receivers on Segment I. According to Lemma \ref{lemB1}, node $a$ must have at least one intended receiver on either Segments III, IV or V. In this case, we can have another assignment, $R$, with $R(i)=R^*(i)$, $\forall i\in \mathcal{N}\setminus\{l_{\text{II}},a\}$, $R(a)=d(a,l_{\text{II}})$,\footnote{We can have $a\dashrightarrow l_{\text{II}}$ to save even more energy. For this, we must have no other nodes with increased range on Segment II, which is proved in Lemma \ref{lemB5}.} and  $R(l_{\text{II}})=\max\{R^*(l_{\text{II}}),d(l_{\text{II}},r^{R^*}_a)\}$, where $r^{R^*}_a$ denotes the farthest intended receiver of node $a$ when it transmits with $R^*(a)$. By using $R$, the data will be sent to all the nodes with less energy, which contradicts the optimality of $R^*$.
\end{proof}

\begin{lemma}
\label{lemB5}
In a cross network, utilizing $R^*$, at most one node in $\mathcal{\hat{N}}$ with increased transmission range exists on Segment II.
\end{lemma}

\begin{proof}
The proof is by contradiction. Suppose in $R^*$ two nodes in $\mathcal{\hat{N}}$ with increased range exist on Segment II. We denote the node closer to the source by $a$, and the node closer to the intersection by $b$. 
We know that $a\overset{R^*}{\nleftarrow}b$ (based on Corollary \ref{cor2}), so according to Lemma \ref{lemB2incd}, we have $b\overset{R^*}{\leftarrow}a$.

According to Lemma \ref{lemBonII}, both nodes $a$ and $b$ must have intended receivers on Segment I. Since $a\overset{R^*}{\nleftarrow}b$, the intended receiver of node $b$ on Segment I is after node $r^{R^*}_{a,\text{I}}$ on Segment I (Lemma \ref{lem5}). This results in having the whole transmission circle of node $a$ being inside the transmission circle of node $b$.
To minimize energy (according to Corollary \ref{cor1}) for delivering data from $a$ to $b$, we can use $a\dashrightarrow r^{R^*}_{a,\text{II}}$ (or $a\dashrightarrow b$, if node $b$ is before node $r^{R^*}_{a,\text{II}}$), while all the other nodes transmit as before. This way all the nodes receive the data, but node $a$ does not have increased range, which contradicts our assumption.
\end{proof}

\begin{lemma}
\label{lemBjust2}
For a cross network utilizing the optimal assignment $R^*$, no more than two nodes with increased range exist on a segment.
\end{lemma}

\begin{proof}
We first study some properties of the case where in $R^*$, two nodes from $\mathcal{\hat{N}}$ on one segment have increased transmission range. Then we show that no other node with increased range in $R^*$ can exist on that segment.

Assume nodes $a$ and $b$ have increased range and node $a$ is before node $b$ on the same segment. According to Lemma \ref{lemB1}, these two nodes must have some intended receivers on other segments. We know that $a\overset{R^*}{\nleftarrow}b$ (based on Corollary \ref{cor2}), so according to Lemma \ref{lemB2incd}, we have $b\overset{R^*}{\leftarrow}a$.

According to Lemma \ref{lem5}, the intended receivers of node $b$ on any segment must be after the last intended receiver of $a$ on that segment. If node $b$ has any intended receivers on the segment aligned to its segment, the whole transmission circle of node $a$ will fall into the range of node $b$. In this case we just need to deliver the data from node $a$ to node $b$, and node $b$ transmits the data to all the other-segment receivers of node $a$. To minimize energy (according to Corollary \ref{cor1}) for delivering data from $a$ to all its same-segment receivers (when it transmits with $R^*(a)$) we can use $a\dashrightarrow r^{R^*}_{a,\mathcal{S}_a}$ (or $a\dashrightarrow b$, if node $b$ is before node $r^{R^*}_{a,\mathcal{S}_a}$, and $b$ covers the rest), while all the other nodes transmit as before.
Note that, according to Corollary \ref{corBnan}, the nodes from $\mathfrak{n}_a$ up to $r^{R^*}_{a,\mathcal{S}_a}$ that do not have increased range (i.e., except for node $b$, if $b$ is before $r^{R^*}_{a,\mathcal{S}_a}$), have zero as their transmission range, and have no effect on the delivery of data to other nodes. This means that all the nodes receive the data while node $a$ does not have increased range, which contradicts our assumption. Therefore, the only possible case is that node $a$ has some intended receivers (beyond the range of $b$) on the aligned segment to its segment, and node $b$ has some intended receivers (beyond the range of $a$) on (one of the) perpendicular segments to its segment (except for Segment II, according to Lemma \ref{lemBII}).

Now, we prove that no other node with increased range exists on segment $\mathcal{S}_a$. We showed in Lemma \ref{lemB5} that no more than one node with increased range exists on Segment II. So, the proof of this lemma for Segment II is already given. Two other cases for the segment of interest remain:

1) Segment I.

2) One of the Segments III, IV or V.

Case (1):
If node $a$ just has intended receivers on Segment II, then node $b$ has an intended receiver on another segment, which results in covering the whole transmission circle of node $a$, which as we showed before, is not possible. This means that node $a$ has some intended receivers on Segments III, IV, or V. The rest of the proof is by contradiction. Suppose nodes $a$, $b$ and $c$, all located on Segment I, have increased range. Without loss of generality, assume that node $c$ is after node $b$, which is after node $a$. According to Corollary \ref{cor2}, $a\overset{R^*}{\nleftarrow}c$ and $b\overset{R^*}{\nleftarrow}c$. For node $c$ to have increased range, according to Lemma \ref{lemB1}, it has to have some intended receivers on other segments rather than $\mathcal{S}_c$. If node $c$ has an intended receiver on one of the other segments, according to Lemma \ref{lem5}, that node is after the last intended receiver of nodes $a$ and $b$ on that segment.
Due to the circular shape of transmission range of nodes, if node $c$ has any intended receivers on other segments, all the intended receivers of node $b$ on other segments will fall into the range of node $c$, and we just need to deliver the data from node $b$ to all its same-segment receivers. To minimize the energy, according to Corollary \ref{cor1}, we can use $b\dashrightarrow r^{R^*}_{b,\mathcal{S}_b}$ (or $b\dashrightarrow c$, if node $c$ is before node $r^{R^*}_{b,\mathcal{S}_b}$), while all the other nodes transmit as before, in which node $b$ does not have increased range. This contradicts our assumption.

Case (2):
According to Lemma \ref{lemBII}, the nodes on Segments III, IV and V do not have intended receivers on Segment II. Using the same approach as in case (1), while ignoring Segment II, the lemma can be proved.
\end{proof}

\begin{lemma}
\label{lemBIInointonI}
Consider a cross network that utilizes the optimal assignment $R^*$. If a node with increased range exists on Segment II, then the nodes on Segments III, IV and V do not have intended receivers on Segment I.
\end{lemma}

\begin{proof}
The proof is by contradiction. Suppose node $a$ with increased range is located on Segment II, and node $b$ on one of the Segments III, IV or V has intended receivers on Segment I. We know that $a\overset{R^*}{\nleftarrow}b$ (due to the circular transmission ranges of the nodes), so according to Lemma \ref{lem5}, the intended receivers of node $b$ on Segment I are after node $r^{R^*}_{a,\text{I}}$. We can have another assignment, $R$, with $R(i)=R^*(i)$, $\forall i\in \mathcal{N}\setminus\{l_{\text{II}},a\}$, $R(a)=d(a,l_{\text{II}})$,\footnote{We can have $a\dashrightarrow l_{\text{II}}$ to save even more energy.} and  $R(l_{\text{II}})=\max\{R^*(l_{\text{II}}),d(l_{\text{II}},r^{R^*}_a)\}$, where $r^{R^*}_a$ denotes the farthest receiver of node $a$ on Segments III, IV and V, when it transmits with $R^*(a)$. By using $R$, the data will be sent to all the nodes with less energy, which contradicts the optimality of $R^*$.
\end{proof}

\begin{lemma}
\label{lemBIIjustoneI}
Consider a cross network that utilizes the optimal assignment $R^*$. If two nodes with increased range exist, each on one of the Segments I and II, then there is no other node with increased range.
\end{lemma}

\begin{proof}
Suppose nodes $a$ and $b$ with increased range are located on Segments I and II, respectively. According to Lemma \ref{lemBonII}, node $b$ has an intended receiver on Segment I. If $a\overset{R^*}{\nleftarrow}b$, then according to Lemma \ref{lem5}, the intended receivers of $b$ on Segment I are after the last same-segment receiver of $a$, which, according to Lemma \ref{lemB4}, contradicts the optimality of $R^*$. So we have $a\overset{R^*}{\leftarrow}b$.

According to Lemma \ref{lem5}, the intended receivers of node $a$ on any segment are after the last intended receiver of $b$ on that segment. If node $b$ does not have any receivers on Segments III, IV and V, or node $a$ has some intended receivers on either Segments II or III, the whole transmission circle of node $b$ will fall into the range of node $a$. In this case, we just need to deliver the data to node $a$, and node $a$ transmits the data to all the receivers of node $b$. Using another transmission range assignment $R$, with $R(i)=R^*(i)$, $\forall i\in \mathcal{N}\setminus\{s,b\}$, $R(b)=0$, and  $R(s)=\max\{R^*(s),d(s,r^{R^*}_{a,\text{I}})\}$, all the nodes receive the data with less energy, which contradicts the optimality of $R^*$. Therefore, the only possible case is that node $b$ has some intended receivers on Segment III, and node $a$ has some intended receivers (beyond the range of $b$) on (one of the) Segments IV and V.

Now suppose that we have a third node with increased range (denoted by $c$). There will be four cases for the location of node $c$:

1) Node $c$ is on Segment I.

2) Node $c$ is on Segment II.

3) Node $c$ is on Segment III.

4) Node $c$ is on either Segments IV or V.

Case (1):
Without loss of generality, we assume that node $c$ is after node $a$. According to Corollary \ref{cor2}, $a\overset{R^*}{\nleftarrow}c$, and since we have $a\overset{R^*}{\leftarrow}b$, so we have $b\overset{R^*}{\nleftarrow}c$. For node $c$ to have increased range, according to Lemma \ref{lemB1}, it has to have some intended receivers on other segments rather than $\mathcal{S}_c$. If node $c$ has an intended receiver on one of the other segments, according to Lemma \ref{lem5}, that node is after the last intended receiver of nodes $a$ and $b$ on that segment. If node $c$ has any intended receivers on another segment, all the intended receivers of node $a$ on other segments will fall into the range of node $c$, and we just need to deliver the data from node $a$ to node $c$. To minimize the energy, according to Corollary \ref{cor1}, we can use $a\dashrightarrow r^{R^*}_{a,\text{I}}$ (or $a\dashrightarrow c$, if node $c$ is before node $r^{R^*}_{a,\text{I}}$), while other nodes transmit as before. According to Lemma \ref{lemB4}, no nodes on Segments II, III, IV, or V have receivers after $r^{R^*}_{a,\text{I}}$. Hence node $c$ receives the data from the nodes on its segment. This way, all the nodes receive the data, while node $a$ does not have increased range. This contradicts our assumption. So this case is not possible.

Case (2):
In Lemma \ref{lemB5}, we showed that this case is not possible.

Before discussing the remaining cases, we prove that for a node $d$ on either Segments IV or V, if we have $d\overset{R^*}{\leftarrow}a$, node $d$ does not cover all the receivers of node $b$ on Segment III.

The proof is by contradiction. Suppose for node $d$, located on one of the Segments IV or V, we have $d\overset{R^*}{\leftarrow}a$, and also all the receivers of $b$ on Segment III are in the range of node $d$. Using another transmission range assignment $R$, with $R(i)=R^*(i)$, $\forall i\in \mathcal{N}\setminus\{s,b\}$, $R(b)=0$, and  $R(s)=\max\{R^*(s),d(s,r^{R^*}_{b,\text{I}})\}$, all the nodes receive the data\footnote{Node $a$ receives and transmits to all the nodes covered by $R^*(a)$ and $R^*(b)$ on Segments IV and V (including $\mathcal{S}_d$). Node $d$ receives and transmits to all the nodes covered by $R^*(b)$ on Segment III. The rest of the nodes receive the data similar to when we use $R^*$.} with less energy, which contradicts the optimality of $R^*$.

Also, note that for any node $c$ on Segments III, IV or V, we have $b\overset{R^*}{\nleftarrow}c$ (Lemma \ref{lemBII}). Since the nodes on Segments III, IV or V do not have intended receivers after $r^{R^*}_{b,\text{I}}$ on Segment I (Lemma \ref{lemBIInointonI}), node $a$ does not receive the data via a path containing a node on Segments III, IV or V.
So for the remaining two cases, we have $b\overset{R^*}{\nleftarrow}c$ and $a\overset{R^*}{\nleftarrow}c$, which according to Lemma \ref{lemB2incd}, results in having $c\overset{R^*}{\leftarrow}b$ and $c\overset{R^*}{\leftarrow}a$.

Case (3):
We have $c\overset{R^*}{\leftarrow}a$. Node $a$ has no intended receivers on Segment III. So, for a node (e.g., $c$) on Segment III, we have $c\overset{R^*}{\leftarrow}a$, if $c$ is after $r^{R^*}_{b,\text{III}}$, and a node (e.g., $d$) on a segment other than Segments I and II exists that $d\overset{R^*}{\leftarrow}a$ and $d$ transmits to the nodes after $r^{R^*}_{b,\text{III}}$. Otherwise, if $c$ is within the range of $b$, we will have $c\overset{R^*}{\nleftarrow}a$. Also, if $c$ is after $r^{R^*}_{b,\text{III}}$ and node $d$ does not exist, we will have $c\overset{R^*}{\leftarrow}e$, where $e$ is a receiver of node $b$ on Segment III, which results in $c\overset{R^*}{\nleftarrow}a$. But we proved that if such node (node $d$ on one of the Segments IV or V, for which we have $d\overset{R^*}{\leftarrow}a$) exists, it does not cover all the receivers of node $b$ on Segment III. So this case is not possible.

Case (4):
We have $c\overset{R^*}{\leftarrow}a$ and $c\overset{R^*}{\leftarrow}b$. According to Lemmas \ref{lemBII} and \ref{lemBIInointonI} node $c$ has no intended receivers on Segments I and II. The other-segment intended receivers of node $c$ on the two other remaining segments are after the last receivers of nodes $a$ and $b$ on those segments, according to Lemma \ref{lem5}. In both cases, all the receivers of node $b$ on Segment III fall into the range of node $c$.\footnote{If $c$ has intended receivers on the perpendicular segment to its segment, we have $cov^{R^*}_{\bot}(c)>cov^{R^*}_{oppo}(b)$. If $c$ has intended receivers on the aligned segment to its segment, we have $cov^{R^*}_{\bot}(c)\geq cov^{R^*}_{oppo}(c)>cov^{R^*}_{\bot}(a)>cov^{R^*}_{\bot}(b)\geq cov^{R^*}_{oppo}(b)$.} We proved that if such a node (node $d=c$ on one of the Segments IV or V, where $(d=c)\overset{R^*}{\leftarrow}b$) exists, it does not cover all the receivers of node $b$ on Segment III. So this case is not possible.
\end{proof}

\begin{proof}[Proof of Theorem \ref{theB1}]
Two cases may happen:

1) A node with increased range does not exist on Segment II.

2) A node with increased range exists on Segment II.

Case (1.a):

First we consider the case where only one node on any segment (except for Segment II) has increased range in $R^*$. Suppose there are four nodes with increased range in $R^*$, and each node is located on one of the Segments I, III, IV and V. Denote these nodes by $a$, $b$, $c$ and $d$. Each of these nodes has to have some intended receivers on other segments than its own segment (Lemma \ref{lemB1}). Without loss of generality, and according to Lemma \ref{lemB2incd}, assume $b\overset{R^*}{\leftarrow}a$, $c\overset{R^*}{\leftarrow}b$ (which results in $c\overset{R^*}{\leftarrow}a$) and $d\overset{R^*}{\leftarrow}c$ (which means $d\overset{R^*}{\leftarrow}a$ and $d\overset{R^*}{\leftarrow}b$). So we have $a\overset{R^*}{\nleftarrow}d$, $b\overset{R^*}{\nleftarrow}d$ and $c\overset{R^*}{\nleftarrow}d$. According to Lemmas \ref{lem5} and \ref{lemB4}, node $d$ does not have intended receivers on $\mathcal{S}_a$, $\mathcal{S}_b$ and $\mathcal{S}_c$.
Since the nodes on Segments III, IV and V do not have intended receivers on Segment II (Lemma \ref{lemBII}), the only possible case is that node $d$ is on Segment I and has some intended receivers on Segment II. This means that if node $d$ does not transmit, some nodes on Segment II will not receive the data. Therefore, all the nodes on Segments II, IV and V will not receive the data (circular transmission of the nodes results in reception of data by all the nodes on Segment II before all the nodes on Segments III, IV and V). This contradicts with $d\overset{R^*}{\leftarrow}a$, $d\overset{R^*}{\leftarrow}b$ and $d\overset{R^*}{\leftarrow}c$.
Therefore, node $d$ does not have any intended receivers on other segments, which according to Lemma \ref{lemB1} contradicts the optimality of $R^*$. So, in this case, at most three nodes with increased range exist in $R^*$.

Case (1.b):

Now assume that in $R^*$, two nodes from $\mathcal{\hat{N}}$ on one segment (except for Segment II) have increased transmission range. Assume nodes $a$ and $b$ have increased range and node $a$ is before node $b$ on the same segment. 
In the proof of Lemma \ref{lemBjust2}, we showed that $b\overset{R^*}{\leftarrow}a$. Also, we showed that the only possible case is that node $a$ has some intended receivers on the aligned segment to its segment,\footnote{If $\mathcal{S}_a=\text{I}$, this aligned segment is just Segment III (see proof of Lemma \ref{lemBjust2}).} and node $b$ has some intended receivers (beyond the range of $a$) on (one of the) perpendicular segments to $\mathcal{S}_a$.

Before further discussion, we prove that for a node $d$ on a segment perpendicular to $\mathcal{S}_a$, if we have $d\overset{R^*}{\leftarrow}b$, node $d$ does not cover all the receivers of node $a$ on a segment aligned to $\mathcal{S}_a$.

The proof is by contradiction. Suppose for node $d$, located on a perpendicular segment to $\mathcal{S}_a$, we have $d\overset{R^*}{\leftarrow}b$, and also all the receivers of $a$ on a segment aligned to $\mathcal{S}_a$ are in the range of node $d$.
Using $a\dashrightarrow r^{R^*}_{a,\mathcal{S}_a}$ (or $a\dashrightarrow b$, if node $b$ is before node $r^{R^*}_{a,\mathcal{S}_a}$), while all the other nodes transmit as before, all the nodes in the network receive the data.\footnote{Node $b$ receives and transmits to all the nodes covered by $R^*(a)$ and $R^*(b)$ on the perpendicular segments to $\mathcal{S}_a$ (including $\mathcal{S}_d$). Node $d$ receives and transmits to all the nodes covered by $R^*(a)$ on a segment aligned to $\mathcal{S}_a$. The rest of the nodes receive the data similar to when we use $R^*$.} This way, node $a$ does not have increased range, which is a contradiction.

Now suppose that we have a third node with increased range (denoted by $c$). There will be four cases for the location of node $c$ (except for Segment II):

1.b.i) Node $c$ is on the same segment as nodes $a$ and $b$, i.e., segment $\mathcal{S}_a$.

1.b.ii) Node $c$ is on the segment aligned to the segment of nodes $a$ and $b$.

1.b.iii) Node $c$ is on a segment perpendicular to the segment of nodes $a$ and $b$, on which node $b$ has intended receivers.

1.b.iv) Node $c$ is on a segment perpendicular to the segment of nodes $a$ and $b$, on which node $b$ does not have any intended receivers.

Case (1.b.i):
This case is not possible, according to Lemma \ref{lemBjust2}.

Case (1.b.ii):
First we prove that $c\overset{R^*}{\leftarrow}a$. Using proof by contradiction, we assume $c\overset{R^*}{\nleftarrow}a$. We know that node $a$ has intended receivers on $\mathcal{S}_c$. According to Lemma \ref{lem5}, the intended receivers of $a$ on $\mathcal{S}_c$ are after the last same-segment receiver of $c$, which, according to Lemma \ref{lemB4}, contradicts the optimality of $R^*$. Thus we have $c\overset{R^*}{\leftarrow}a$.

Now we prove that $b\overset{R^*}{\nleftarrow}c$. Suppose node $c$ does not transmit. Since we have $a\overset{R^*}{\nleftarrow}c$, node $a$ receives the data, and transmits it. If node $b$ is within the range of node $a$, then $b$ will receive the data as well, and obviously we have $b\overset{R^*}{\nleftarrow}c$.
According to Lemma \ref{lemB4}, no nodes on other segments have receivers after $r^{R^*}_{a,\mathcal{S}_a}$.
If node $b$ is after $r^{R^*}_{a,\mathcal{S}_a}$, it still receives the data via the nodes on its segment. So, in this case also, we have $b\overset{R^*}{\nleftarrow}c$.

Since $b\overset{R^*}{\nleftarrow}c$, based on Lemma \ref{lemB2incd} we have $c\overset{R^*}{\leftarrow}b$. This implies node $c$ is not within the range of $a$. Node $b$ has no intended receivers on $\mathcal{S}_c$. So, for node $c$ we have $c\overset{R^*}{\leftarrow}b$, if $c$ is after $r^{R^*}_{a,\mathcal{S}_c}$, and a node (e.g., $d$) on a segment other than $\mathcal{S}_a$ and $\mathcal{S}_c$ exists that $d\overset{R^*}{\leftarrow}b$ and $d$ transmits to the nodes after $r^{R^*}_{a,\mathcal{S}_c}$. If $c$ is after $r^{R^*}_{a,\mathcal{S}_c}$ and such node $d$ does not exist, we will have $c\overset{R^*}{\leftarrow}e$, where $e$ is a receiver of node $a$ on $\mathcal{S}_c$, which results in $c\overset{R^*}{\nleftarrow}b$. Therefore, node $d$ on one of the perpendicular segments to $\mathcal{S}_a$ and $\mathcal{S}_c$ exists, and we have $d\overset{R^*}{\leftarrow}b$. Node $d$ covers all the receivers of node $a$ on segment $\mathcal{S}_c$. This contradicts what we proved before studying the cases. Therefore, this case is not possible.

Case (1.b.iii):
First we prove that $c\overset{R^*}{\leftarrow}b$. Using proof by contradiction, we assume $c\overset{R^*}{\nleftarrow}b$. We know that node $b$ has intended receivers on $\mathcal{S}_c$. According to Lemma \ref{lem5} the intended receivers of $b$ on $\mathcal{S}_c$ are after the last same-segment receiver of $c$, which, according to Lemma \ref{lemB4}, contradicts the optimality of $R^*$. So we have $c\overset{R^*}{\leftarrow}b$.

Knowing $b\overset{R^*}{\nleftarrow}c$, and using Corollary \ref{cor3}, we have $a\overset{R^*}{\nleftarrow}c$. So, according to Lemma \ref{lemB2incd}, $c\overset{R^*}{\leftarrow}a$. If node $c$ has an intended receiver on $\mathcal{S}_a$, according to Lemma \ref{lem5}, it transmits to a node after the last same-segment receiver of $a$, which, according to Lemma \ref{lemB4}, contradicts the optimality of $R^*$. So, node $c$ has no intended receiver on $\mathcal{S}_a$. The intended receivers of node $c$ on the other remaining segments are after the last receivers of nodes $a$ and $b$ on those segments, according to Lemma \ref{lem5}. This results in having all the receivers of node $a$ on the segments aligned to $\mathcal{S}_a$ to fall into the range of node $c$.\footnote{If $c$ has intended receivers on the perpendicular segment to its segment, we have $cov^{R^*}_{\bot}(c)>cov^{R^*}_{oppo}(a)$. If $c$ has intended receivers on the aligned segment to its segment, we have $cov^{R^*}_{\bot}(c)\geq cov^{R^*}_{oppo}(c)>cov^{R^*}_{\bot}(b)>cov^{R^*}_{\bot}(a)\geq cov^{R^*}_{oppo}(a)$.} Similar to the Case (1.b.ii), if such node (node $d=c$ on one of the perpendicular segments to $\mathcal{S}_a$, where $(d=c)\overset{R^*}{\leftarrow}b$) exists, it does not cover all the receivers of node $a$ on the segments aligned to $\mathcal{S}_a$. So this case is not possible.

Case (1.b.iv):
First we prove that $c\overset{R^*}{\nleftarrow}a$. Using proof by contradiction, we assume $c\overset{R^*}{\leftarrow}a$, and thus $a\overset{R^*}{\nleftarrow}c$. If node $c$ does not transmit, node $a$ receives the data, and transmits it.
Using the same approach as in case (1.2), we have $b\overset{R^*}{\nleftarrow}c$. So, according to Lemma \ref{lemB2incd}, we have $c\overset{R^*}{\leftarrow}b$. Similar to the previous case, node $c$ can not have increased range, as it does not have any intended receivers on other segments.

So, Case (1.b.iv) requires $c\overset{R^*}{\nleftarrow}a$. According to Lemma \ref{lemB2incd}, we have $a\overset{R^*}{\leftarrow}c$, which based on Lemma \ref{lem3} results in $b\overset{R^*}{\leftarrow}c$.

Now, using contradiction, we prove that no fourth node with increased range can exist. Suppose a fourth node with increased range, denoted by $d$, exists. We showed in the previous cases that node $d$ can not be on $\mathcal{S}_a$, the segment aligned to $\mathcal{S}_a$ or the segment perpendicular to $\mathcal{S}_a$ on which node $b$ has some intended receivers, while we have two nodes with increased range on $\mathcal{S}_a$. It cannot be on Segment II as well. Therefore, node $d$ is on the same segment as node $c$ ($\mathcal{S}_c$). Rename the nodes $c$ and $d$ so that node $d$ be after node $c$. Similar to the approach we used to prove $c\overset{R^*}{\nleftarrow}a$, we prove that $d\overset{R^*}{\nleftarrow}a$. Hence according to Lemmas \ref{lem3} and \ref{lemB2incd}, we have $a\overset{R^*}{\leftarrow}d$ and $b\overset{R^*}{\leftarrow}d$. Also we know that $c\overset{R^*}{\nleftarrow}d$ (Corollary \ref{cor2}). The intended receivers of node $d$ on any segment are after the last intended receiver of $c$ on that segment. Similar to the discussion we had for nodes $a$ and $b$, here the only possible case is that node $d$ has intended receivers on one the perpendicular segments to $\mathcal{S}_d$. If we use $c\dashrightarrow r^{R^*}_{c,\mathcal{S}_c}$ (or $c\dashrightarrow d$, if node $d$ is before node $r^{R^*}_{c,\mathcal{S}_c}$), while all the other nodes use the same ranges, all the nodes receive the data\footnote{All the receivers of $c$ on the perpendicular segments to $\mathcal{S}_c$ are covered by $d$. The receivers of $c$ on the aligned segments to $\mathcal{S}_c$ are covered by $d$ (if the aligned segment is Segment II) or $b$ (which receives the data from $d$). That's because according to Lemma \ref{lem5}, the intended receivers of node $b$ on those segments are after the last intended receiver of $c$ on them.} with less energy. This contradicts the optimality of $R^*$.

Case (2):
If a node with increased range exists on Segment I, then according to Lemma \ref{lemBIIjustoneI}, no more nodes with increased range exist. For the rest of the proof, we assume that no node with increased range exists on Segment I.

First, consider the case where only one node on each of the Segments III, IV and V has increased range in $R^*$. The proof is by contradiction. Suppose there are four nodes with increased range in $R^*$, and each node is located on one of the Segments II, III, IV and V. Denote these nodes by $a$, $b$, $c$ and $d$, where node $a$ is located on Segment II. According to Lemmas \ref{lemBII} and \ref{lemBIInointonI}, nodes $b$, $c$ and $d$ do not have intended receivers on Segments I and II. Without loss of generality, and according to Lemma \ref{lemB2incd}, assume $c\overset{R^*}{\leftarrow}b$, $d\overset{R^*}{\leftarrow}c$ (which results in $d\overset{R^*}{\leftarrow}b$). So we have $b\overset{R^*}{\nleftarrow}d$ and $c\overset{R^*}{\nleftarrow}d$. According to Lemmas \ref{lem5} and \ref{lemB4}, node $d$ does not have intended receivers on $\mathcal{S}_b$ and $\mathcal{S}_c$. Therefore, node $d$ does not have any intended receivers on other segments, which according to Lemma \ref{lemB1} contradicts the optimality of $R^*$. So, in this case, at most three nodes with increased range exist in $R^*$.

Now, we consider the case where more than one node with increased range may exist on each of the Segments III, IV and V. According to Lemma \ref{lemBjust2}, no more than two nodes with increased range exist on each segment. Also, according to Lemma \ref{lemB5}, we have at most one node with increased range on Segment II. So, assume nodes $a$ and $b$ on one of the Segments III, IV or V have increased range, while we have a node with increased range on Segment II. We assume that node $a$ is before $b$. Hence, according to Corollary \ref{cor2} and Lemma \ref{lemB2incd}, we have $b\overset{R^*}{\leftarrow}a$. If these two nodes are on Segment III, as they can not have intended receivers on Segments I and II (Lemmas \ref{lemBII} and \ref{lemBIInointonI}), they both have intended receivers on Segments IV and V. Node $b$ having intended receivers on (one of the) Segments IV and V results in having the whole transmission circle of $a$ within range $b$, which means that node $a$ does not need to have increased range. This contradicts our assumption, so this case can not happen. If nodes $a$ and $b$ are on one of the Segments IV or V (e.g., Segment IV), as they can not have intended receivers on Segments I and II (Lemmas \ref{lemBII} and \ref{lemBIInointonI}), they have intended receivers on Segments III and V. Similar to the approach used in cases (1.1), (1.2) and (1.3), we can prove that no other node with increased range exists. So, in this case, at most three nodes with increased range exist in $R^*$.
\end{proof}

To find $R^*$, one have to search among all the possible assignments, constructed from all the possible combinations of the nodes with increased transmission range and the values of their transmission ranges. 
Based on Theorem \ref{theB1}, one needs to  consider all the possible combinations of three nodes chosen from $\hat{\cal{N}}$, i.e., $\binom{N-5}{3}$ choices, and for each such node, all the $N$ possible range values including zero. For the nodes in the set $\{l_{\text{II}},f_{\text{III}},f_{\text{IV}},f_{\text{V}}\}$ too, each node can take any of the $N$ possible range values. For the source node $s$, all the non-zero range values can be selected. To describe the eight nodes with possible increased range, we first define set $\mathcal{T}$ to contain all the possible choices for the three nodes from set $\mathcal{\hat{N}}$, and concatenate each of the possible choices for the three nodes with the nodes $l_{\text{II}}$, $f_{\text{III}}$, $f_{\text{IV}}$, and $f_{\text{V}}$. Therefore, each member $\mathbf{t}\in\mathcal{T}$ is denoted by $\{t_1,t_2,t_3,t_4,t_5,t_6,t_7\}$, where $t_1, t_2$ and $t_3$ represent the three chosen nodes, and $t_4$, $t_5$, $t_6$, and $t_7$ denote the nodes $l_{\text{II}}$, $f_{\text{III}}$, $f_{\text{IV}}$, and $f_{\text{V}}$, respectively.
Set $\mathcal{T}$ has (at most) $\binom{N-5}{3}$ members.
To search all the possible ranges for each of the three selected nodes from $\mathcal{\hat{N}}$, and also the nodes in the set $\{s,l_{\text{II}},f_{\text{III}},f_{\text{IV}},f_{\text{V}}\}$, we construct the set $\mathcal{C}=\mathcal{N}^8$. Each member of $\mathcal{C}$, which is denoted by $\mathbf{c}$, is considered as an 8-tuple of form $(c_0,c_1,c_2,c_3,c_4,c_5,c_6,c_7)$. The range of the source node is equal to $d(s,c_0)$, and the range of node $t_j$ (for $1\leq j\leq 7$) is equal to $d(t_j,c_j)$.

To account for the order in which nodes on different segments transmit data, we also need to search among all the $5!=120$ different segment orderings. We define set $\mathcal{P}$ to contain all the possible permutations of the Segments \{I,II,III,IV,V\}. Each segment ordering $\mathbf{p}=(p_1,p_2,p_3,p_4,p_5)\in \mathcal{P}$ contains the labels of the segments in the order they have to be checked. On each segment, we assign ranges to the nodes starting from the first node up to the last node on that segment. For the source node and the seven nodes from set $\mathcal{T}$, the range is already assigned as discussed before. For each remaining node, we assign its $M$ value to the range, if the next adjacent neighbor of the node has not received the data yet, otherwise, zero is assigned as the range (based on Corollary \ref{corBnan}).

Since all possible assignments do not result in the delivery of data to all the nodes, we need to construct the desired assignments in a way that the delivery of data to all the nodes is guaranteed. This is done by using received labels (label $\mathfrak{r}$). If a node receives the data, we tag it by label $\mathfrak{r}$. An assignment in which a node without this label exists, will be ignored.
The optimal assignment, is the assignment among all the constructed assignments which has the minimum cost.

To reduce the complexity of the algorithm, we define $N^2$ sets $\mathcal{R}_{i,j}, \forall i,j\in \mathcal{N}$. Each set $\mathcal{R}_{i,j}$ contains all the nodes that are within the transmission range of node $i$, when $R(i)=d(i,j)$. Construction of each of these sets can be done in $\mathcal{O}(N)$ time. Hence, all these sets can be acquired with time complexity $\mathcal{O}(N^3)$. Note that the construction of these sets is performed before the execution of the algorithm, and thus does not introduce any additional complexity to the algorithm.

The pseudo code of the algorithm that finds the optimal transmission range assignment (i.e., $R^*$) is given as Algorithm 1.

\vskip 1em
\hrule
\vskip 0.1em
\textbf{Algorithm 1} Optimal Range Assignment
\vskip 0.1em
\hrule

\algsetup{
linenosize=\footnotesize,
linenodelimiter=:,
indent=1em
}


\begin{algorithmic}[1]
\small
\setstretch{1}
\STATE assign $cost^*=+\infty$

\STATE construct sets $\mathcal{P}$, $\mathcal{T}$, $\mathcal{C}$ and $\mathcal{R}_{i,j}, \forall i,j\in \mathcal{N}$

\FORALL {$\mathbf{t}\in \mathcal{T}$}

\FORALL {$\mathbf{c}\in \mathcal{C}$}

\FORALL {$\mathbf{p}\in \mathcal{P}$}

\STATE clear transmission range assignment $R$

\STATE erase the $\mathfrak{r}$ labels of all the nodes

\STATE tag node $s$ with label $\mathfrak{r}$

\STATE assign $R(s)=d(s,c_0)$

\STATE calculate $cost(R)=R^\alpha(s)$

\STATE tag all the nodes in set $\mathcal{R}_{s,c_0}$ with label $\mathfrak{r}$

\FOR {$k=1$ \TO $5$}

\FORALL {nodes $n$ on segment $p_k$ (from the first node to the last node)}

\IF {node $n$ has not received the data yet (i.e., does not have label $\mathfrak{r}$)}

\STATE \textbf{stop} and go to next $\mathbf{p}$

\ELSIF {$n=t_j$ for $1\leq j\leq 7$}

\STATE assign $R(n)=d(n,c_j)$

\STATE tag all the nodes in set $\mathcal{R}_{n,c_j}$ with label $\mathfrak{r}$

\ELSIF {node $\mathfrak{n}_{n}$ exists and does not have label $\mathfrak{r}$}

\STATE assign $R(n)=M(n)$

\STATE tag all the nodes in set $\mathcal{R}_{n,\mathfrak{n}_{n}}$ with label $\mathfrak{r}$

\ELSE

\STATE assign $R(n)=0$

\ENDIF

\STATE update $cost(R)=cost(R)+R^\alpha(n)$

\ENDFOR

\ENDFOR

\IF {all the nodes receive the data \AND $cost(R)<cost^*$}

\STATE assign $cost^*=cost(R)$ and $R^*=R$

\ENDIF

\ENDFOR

\ENDFOR

\ENDFOR
\end{algorithmic}

\hrule
\vskip 1.7em

\begin{theorem}
\label{theB2}
Algorithm 1 finds $R^*$ with time complexity $\mathcal{O}(N^{12})$.
\end{theorem}

\begin{proof}
There are $\mathcal{O}(N^{11})$ choices for the selection of nodes in the set $\mathcal{T}$, and range assignments from $\mathcal{C}$. For each such choice, there are $N-8$ remaining nodes whose range is assigned as either zero or their $M$ value.
\end{proof}

\subsection{Near-Optimal Range Assignment With Linear Complexity}
The focus of the proposed sub-optimal algorithm is to find the nodes that can save energy by not transmitting, whereas the other nodes only transmit at a power level that is needed for their next adjacent neighbors to receive the data. As we will find out later, this algorithm performs close to optimal. In the following, we thus refer to it as being ``near-optimal''.

Similar to the optimal algorithm, and to check all the segment orderings, we define set $\mathcal{P}$ to contain all the $120$ possible permutations of the Segments \{I,II,III,IV,V\}.

For any node $n$ on each of the ordered segments, we assign $R(n)=M(n)$, if node $\mathfrak{n}_n$ has not received the data through the previously assigned transmission ranges. Otherwise, we assign $R(n)=0$. Here, unlike the optimal algorithm, and to reduce complexity, we do not construct sets $\mathcal{R}_{i,j}$. Instead, we just find the receivers of a fixed number of nodes (not a function of $N$). This has complexity $\mathcal{O}(N)$.

To guarantee the delivery of data to all the nodes, we check the first node of Segments III, IV and V to see if it receives the data or not. If not, we change the range of node $sn$\footnote{See line 30 of Algorithm 2 for definition.} in the set $\{s,l_{\text{II}},f_{\text{III}},f_{\text{IV}},f_{\text{V}}\}$ so that it delivers the data to the first node of the segment of interest.
If Segment II is empty, we assume that node $s$ takes all the functionalities of node $l_{\text{II}}$.

The near-optimal algorithm is described in Algorithm 2. It has output $R^{NO}$ as the near-optimal range assignment.

\vskip 1em
\hrule
\vskip 0.1em
\textbf{Algorithm 2} Near-Optimal Linear-Time Algorithm
\vskip 0.1em
\hrule

\algsetup{
linenosize=\footnotesize,
linenodelimiter=:,
indent=1em
}

\begin{algorithmic}[1]
\small
\setstretch{1}

\STATE assign $cost^{NO}=+\infty$

\STATE construct set $\mathcal{P}$

\FORALL {$\mathbf{p}\in \mathcal{P}$}

\STATE clear transmission range assignment $R$

\STATE erase the $\mathfrak{r}$ labels of all the nodes

\STATE tag node $s$ with label $\mathfrak{r}$

\STATE assign $R(s)=d(s,f_{p_1})$, where $f_{p_1}$ denotes the first node on segment $p_1$

\STATE calculate $cost(R)=R^\alpha(s)$

\STATE tag all the receivers of node $s$ with label $\mathfrak{r}$

\FOR {$k=1$ \TO $5$}

\FORALL {nodes $n$ on segment $p_k$ (from the first node to the last node)}

\IF {node $n$ has not received the data yet (i.e., is not tagged by label $\mathfrak{r}$)}

\STATE \textbf{stop} and go to next $\mathbf{p}$

\ELSIF {node $\mathfrak{n}_n$ exists and does not have label $\mathfrak{r}$}

\STATE assign $R(n)=M(n)$

\STATE tag node $\mathfrak{n}_n$ with label $\mathfrak{r}$

\ELSE

\STATE assign $R(n)=0$

\ENDIF

\STATE update $cost(R)=cost(R)+R^\alpha(n)$

\ENDFOR

\STATE find $m^{p_k}_{\bot}=\underset{\text{$m$ on segment $p_k$}}{\arg \max}\{cov^{R}_{\bot}(m)\}, m^{p_k}_{oppo}=\underset{\text{$m$ on segment $p_k$}}{\arg \max}\{cov^{R}_{oppo}(m)\}$

\IF {$p_k=\text{I}$}

\STATE find $m^{\text{I}}_{\text{II}}=\underset{\text{$m$ on Segment I}}{\arg \max} \{R(m)-d(s,m) \}$\footnote{Node $m^{\text{I}}_{\text{II}}$ denotes the node among all the nodes on Segment I, that has the maximum coverage on Segment II.}

\ELSIF {$p_k=\text{II}$}

\STATE find $m^{\text{II}}_{\text{I}}=\underset{\text{$m$ on Segment II}}{\arg \max} \{R(m)-d(s,m) \}$\footnote{Node $m^{\text{II}}_{\text{I}}$ denotes the node among all the nodes on Segment II, that has the maximum coverage on Segment I.}

\ENDIF

\STATE tag all the nodes that receive the data from nodes $m^{p_k}_{\bot}$, $m^{p_k}_{oppo}$, $m^{\text{I}}_{\text{II}}$, or $m^{\text{II}}_{\text{I}}$ (if they exist) with label $\mathfrak{r}$

\IF {$p_k\neq \text{I}$ \AND $p_{k+1}\in \{\text{III,IV,V}\}$ \AND node $f_{p_{k+1}}$ has not received the data yet}

\STATE reassign $R(sn)=d(sn,f_{p_{k+1}})$, where
$sn=\left\{
      \begin{array}{ll}
        l_{\text{II}}, & \hbox{if $p_k=\text{II}$} \\
        s, & \hbox{if $p_k=\text{II}$, and Segment II is empty} \\
        f_{p_k}, & \hbox{otherwise}
      \end{array}
    \right.
$

\STATE tag all the nodes that receive the data from node $sn$ with label $\mathfrak{r}$

\IF {$p_k \in\{\text{III,IV,V}\}$}

\STATE reassign $R(i)=0$ for all nodes $i\neq sn$ on segment $p_k$, for which node $\mathfrak{n}_i$ is a receiver of node $sn$

\ENDIF

\ENDIF

\ENDFOR

\IF {all the nodes receive the data \AND $cost(R)<cost^{NO}$}

\STATE assign $cost^{NO}=cost(R)$ and label $R^{NO}=R$

\ENDIF

\ENDFOR

\end{algorithmic}

\hrule
\vskip 1.7em

\subsection{Distributed Range Assignment}
In the distributed algorithm, every node in set $\mathcal{\hat{N}}$ just needs to know the distance to its next adjacent neighbor. In this algorithm, the source node transmits the data to its two adjacent neighbors in Segments I and II (we will discuss the case where Segment II is empty later), by having a transmission range sufficient to reach the farthest one.
Any node $a\in \mathcal{\hat{N}}$ waits to receive the data for the first time, then transmits the data to its next adjacent neighbor, i.e., with range $M(a)$.

If $a\in \{l_{\text{II}},f_{\text{III}},f_{\text{IV}},f_{\text{V}}\}$, in addition to its next adjacent neighbor, node $a$ has to consider the other three nodes in this set as well.
The graph consisting of the nodes in the set $\{l_{\text{II}},f_{\text{III}},f_{\text{IV}},f_{\text{V}}\}$, and all the possible edges between them, is denoted by $G^{\diamondsuit}$, and is illustrated in Fig. \ref{fig5}.\footnote{If Segment II is empty, node $l_{\text{II}}$ will be replaced by node $s$.} We refer to this graph as the {\em diamond graph}.

\begin{figure}[h!]
  \centering
    \includegraphics[trim = 10mm 10mm 10mm 10mm, clip, width=3in]{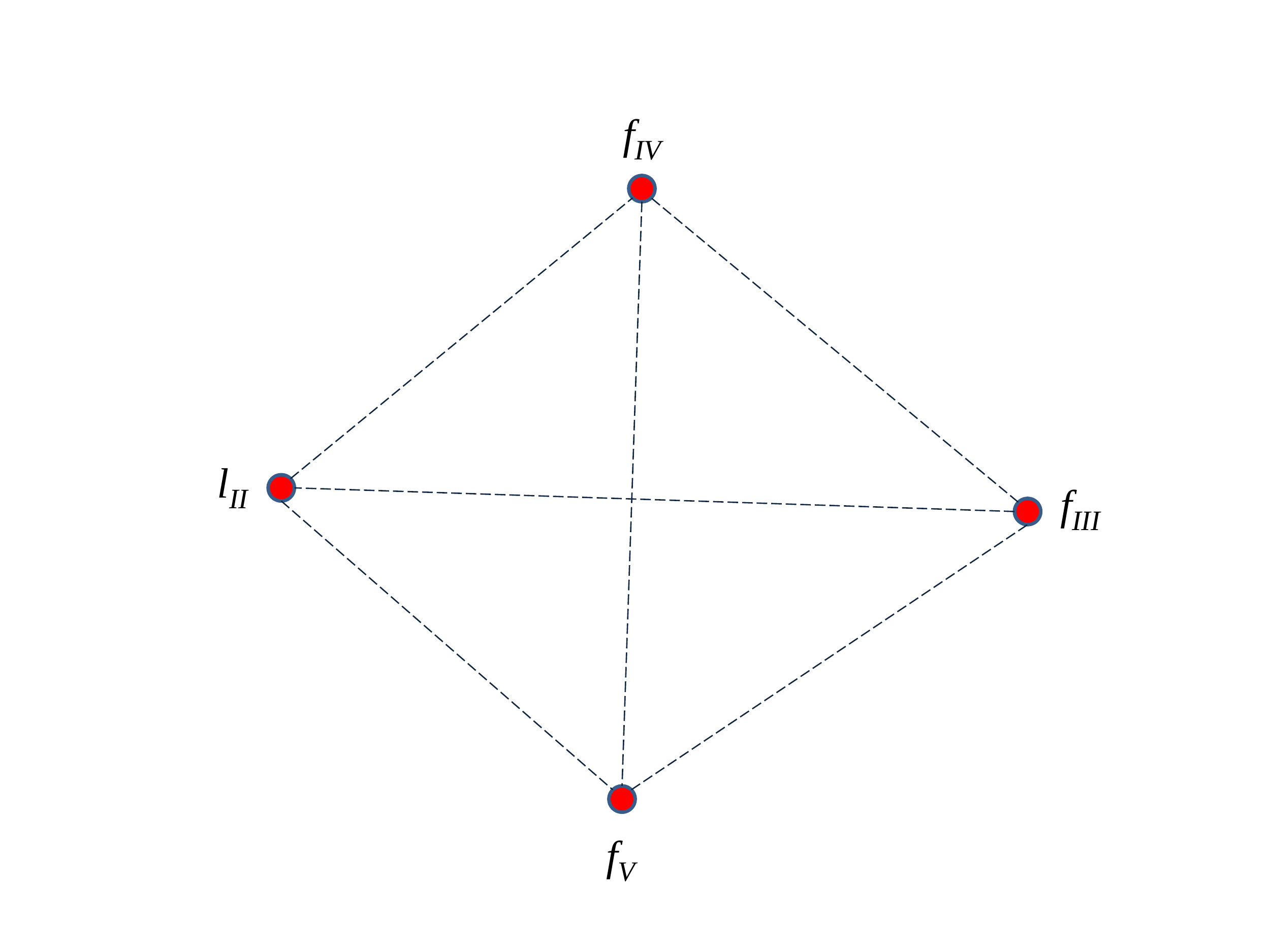}
     \caption{The graph $G^{\diamondsuit}$, for which its MST is needed to be found in each of the nodes in the set $\{l_{\text{II}},f_{\text{III}},f_{\text{IV}},f_{\text{V}}\}$.}
     \label{fig5}
\end{figure}

All the nodes in the set $\{l_{\text{II}},f_{\text{III}},f_{\text{IV}},f_{\text{V}}\}$ calculate the MST of the graph $G^{\diamondsuit}$, e.g., using Prim's algorithm. The transmission range of each of the nodes in the set $\{l_{\text{II}},f_{\text{III}},f_{\text{IV}},f_{\text{V}}\}$ is then selected as the largest of the maximum length of the node's connected edges in the MST of $G^{\diamondsuit}$ and its distance to its next adjacent neighbor. If Segment II is empty, node $s$ replaces node $l_{\text{II}}$ and its range is the largest of the maximum length of its connected edges in the MST of $G^{\diamondsuit}$ and its distance to node $f_{\text{I}}$.

The distributed range assignment, denoted by $R^{D}$, is thus given by:

\begin{align}
\label{RdistB}
R^{D}(s)&=\left\{
            \begin{array}{ll}
              \max\{d(s,f_{\text{I}}),d(s,f_{\text{II}})\}, & \hbox{if Segment II is not empty;} \\
              \max\{d(s,f_{\text{I}}),\max_{u:(s,u)\in MST(G^{\diamondsuit})}\{d(s,u)\}\}, & \hbox{if Segment II is empty,}
            \end{array}
          \right.\\\nonumber
R^{D}(a)&=\max\{M(a),\max_{u:(a,u)\in MST(G^{\diamondsuit})}\{d(a,u)\}\}, \hskip 4em \text{ for } a\in \{l_{\text{II}},f_{\text{III}},f_{\text{IV}},f_{\text{V}}\}, \\\nonumber
R^{D}(a)&= M(a), \hskip 17em \text{ otherwise}.
\end{align}

It is easy to see that the distributed range assignment results in all the network nodes receiving the data. The following theorem shows that the proposed distributed algorithm, with complexity only $\mathcal{O}(1)$, results in the same range assignment as the algorithm of \cite{Nguyen}, with complexity $\mathcal{O}(N^2)$.

\begin{theorem}
\label{theB5}
The distributed transmission range assignment $R^{D}$, is the same as the transmission range assignment of the MST-based algorithm of \cite{Nguyen} for the cross network.
\end{theorem}

\begin{proof}
First we find the MST of the graph corresponding to a cross network. The graph consists of the nodes, as its vertices, and there exists an edge between any two nodes of the network, with weight equal to the distance between them. We assume that the weights of the edges are different, and the MST is unique. In the following, we use the cut\footnote{A cut $[\mathcal{V},\overline{\mathcal{V}}]$ is a partition of the vertices of a graph into two disjoint subsets $\mathcal{V}$ and $\overline{\mathcal{V}}$ that are joined by at least one edge \cite{graph_book}.} property of MST, which states that for any cut in the graph, if the weight of an edge crossing the two partitions of the cut is strictly smaller than the weights of all other crossing edges, then this edge belongs to the MST of the graph (for more detail, see, e.g., Theorem 23.1 of \cite{algorithms_book}). A direct corollary of the cut property (described as Corollary 23.2 in \cite{algorithms_book}) is that if we partition a graph into disjoint sets of vertices, the edge with the minimum weight among all the crossing edges (in the original graph) between any two disjoint partitions belongs to the MST of the graph. So one way to construct the MST of a graph is to partition it into disjoint vertex sets, find the MST of the subgraph induced by each set, and link the MST's via the minimum weight crossing edge between them.

We partition the network graph into seven (or five, if Segment II is empty) partitions. The set of partitions consists of the source node alone as one partition, graph $G^{\diamondsuit}$as another partition, and five (or four) other partitions each consisting of just the remaining nodes of each segment (i.e., one partition for each segment). If Segment II is empty, node $s$ will be a node in $G^{\diamondsuit}$, and there will be only five partitions.

The MST of the $s$ only partition is $s$ itself, and we already discussed the MST of the $G^{\diamondsuit}$ partition. For any node $a\in \mathcal{\hat{N}}$, we define set $\mathcal{V}_a$ to contain node $a$ and all the nodes after that on segment $\mathcal{S}_a$. Set $\overline{\mathcal{V}}_a$ contains all the nodes (in set $\mathcal{\hat{N}}$) on segment $\mathcal{S}_a$, that are not in set $\mathcal{V}_a$. Consider cut $[\mathcal{V}_a,\overline{\mathcal{V}}_a]$ for the subgraph of segment $\mathcal{S}_a$. Since all the vertices of this subgraph are on a straight line, the minimum weight (distance) edge crossing $[\mathcal{V}_a,\overline{\mathcal{V}}_a]$ is between nodes $a$ and $b$, where $\mathfrak{n}_b=a$. This edge, according to the cut property of MST, belongs to the MST of the network graph. Using the same approach for all the cuts of a subgraph, we find all the connecting edges, which all together construct a spanning tree of the subgraph. All these edges have to be included in the MST of each subgraph, and no other edge is necessary for forming a spanning tree. By connecting the different partitions of the graph using the minimum weight (distance) edge between any two partitions the MST of the network graph is obtained. The connecting edge of partition of Segment I and partition of $s$ is the edge between $s$ and $f_{\text{I}}$. Also, the edge between $s$ and $f_{\text{II}}$ is the edge between partition $s$ and partition of Segment II (if Segment II is not empty). The three remaining segment partitions are connected to the partition of $G^{\diamondsuit}$ by the edges between the first nodes of each segment (members of partition $G^{\diamondsuit}$) and the second nodes of the same segments (each member of a different partition). We root the tree at $s$.

In the MST-based transmission range assignment, every node transmits with the range equal to the maximum edge weight (distance) to its children. By observing the way that the MST of the network graph is constructed, we can see that the MST-based transmission range assignment is exactly the same as the distributed transmission range assignment.
\end{proof}

\section{Special Case: Source at the Intersection}
\label{sec:atintersect}
In this section, we consider cross networks in which the source node is located at the intersection of the two lines.\footnote{The work in this section is described with more details in \cite{ataei_icc_cross}.} In such networks, we assume that Segment II still exists, but is empty. Hence, node $s$ takes all the functionalities of node $l_{\text{II}}$. Furthermore, all the nodes in the set $\{l_{\text{II}},f_{\text{III}},f_{\text{IV}},f_{\text{V}}\}$ are replaced by the source node $s$. By doing this, we can see that all the lemmas, corollaries and Theorem \ref{theB1} presented in Section \ref{optBsection} are valid for this special-case network.
Since, for this network, five nodes ($s$, $l_{\text{II}}$, $f_{\text{III}}$, $f_{\text{IV}}$, $f_{\text{V}}$), each with $N$ possible ranges are substituted by only one such node ($s$), the search space will be reduced by a factor of $N^4$. This results in having the time complexity of $\mathcal{O}(N^{8})$.

\section{A more general case: Grid Networks}
\label{sec:grid}
We can use the proposed distributed algorithm to find a cost-efficient transmission range assignment for grid networks with perpendicular line-segments. The structure of the grid can be arbitrary, e.g., the one shown in Fig. \ref{fig6}.

\begin{figure}[h!]
  \centering
    \includegraphics[trim = 10mm 40mm 20mm 40mm, clip, width=5in]{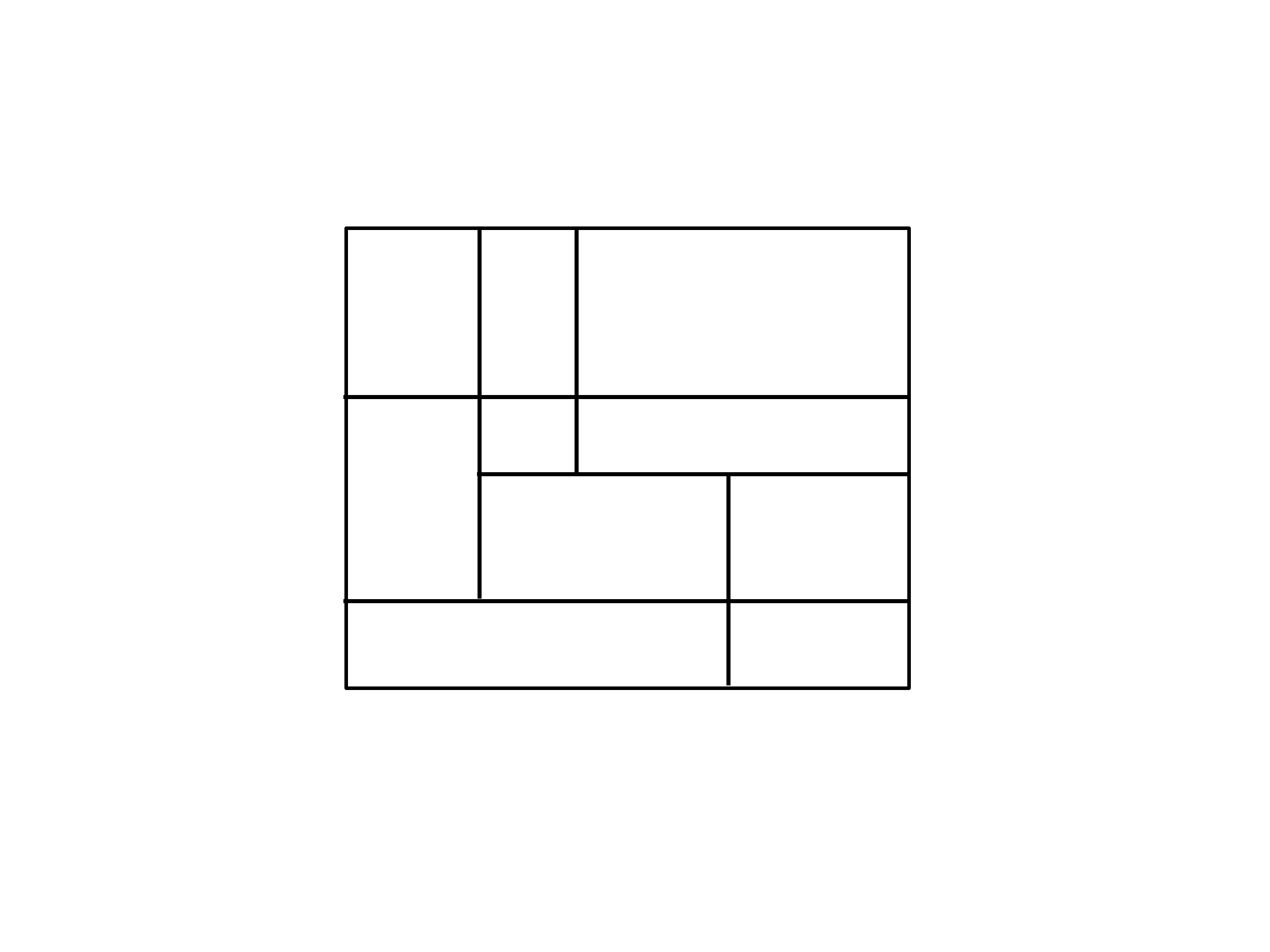}
     \caption{A possible structure for grid networks with perpendicular line-segments.}
     \label{fig6}
\end{figure}

To apply the proposed distributed algorithm, we assume that there is at least one node on each line-segment of the network. This allows for the construction of diamond graphs for each intersection of the network. One can then run the distributed algorithm on each intersection of the grid network similar to what was done for cross networks. In Section \ref{sec:num}, we compare the distributed algorithm to BIP (with and without the sweeping procedure), for a two by two square grid network, and show that the performance difference between the two algorithms is rather small, even though our algorithm is distributed and with complexity $\mathcal{O}(1)$, while BIP is centralized and with complexity $\mathcal{O}(N^2)$.

\section{Numerical Results}
\label{sec:num}
We study the performance of our proposed algorithms, and the BIP algorithm, by conducting Monte Carlo simulations for different number of nodes on cross networks. We also perform the sweep procedure on the BIP algorithm. But since this procedure has time complexity $\mathcal{O}(N^2)$, to keep the complexity of the proposed near-optimal and distributed assignments as low as possible, we do not perform the sweep procedure on these assignments. The complexity of the optimal algorithm is high.
We thus present our results for two relatively small values of $N$, i.e., $N=14$ and $N=18$, in the case where the source node is at the intersection, and for $N=13$ in the general case, where the source node and its location are chosen randomly. For each value of $N$, we simulate 100 networks. For each network, the nodes are distributed uniformly at random on the cross. The simulation results are summarized in Tables \ref{tableA} and \ref{tableB} for the case where source is at the intersection and the general case, respectively. To make the comparisons easier, we normalize the costs by dividing them by the cost of the optimal assignment in each case. The numbers shown in this table are the average of the normalized costs plus minus the $95\%$ confidence interval. The cost of the BIP algorithm and the BIP algorithm with the sweeping procedure are denoted by $cost(R^{BIP})$ and $cost(R^{BIP/sw})$, respectively. In all simulations, the value of $\alpha$ is selected to be 2.

\begin{table}[h!]
\centering
\caption{The simulation results of cross networks with source at intersection.}
\label{tableA}
\begin{tabular}{|c|c|c|c|c|}
  \hline
   & $\frac{cost(R^{NO})}{cost(R^*)}$ & $\frac{cost(R^{BIP/sw})}{cost(R^*)}$ & $\frac{cost(R^{BIP})}{cost(R^*)}$ & $\frac{cost(R^{D})}{cost(R^*)}$ \\
   \hline
  $N=14$ & 1.1140 $\pm$ 0.0276 &  1.2244 $\pm$ 0.0364 & 1.3009 $\pm$ 0.0391 & 1.4303 $\pm$ 0.0554 \\
  \hline
  $N=18$ & 1.1102 $\pm$ 0.0251 & 1.2100 $\pm$ 0.0338 &1.2623 $\pm$ 0.0360 & 1.3666 $\pm$ 0.0412 \\
  \hline
\end{tabular}
\end{table}

\begin{table}[h!]
\centering
\caption{The simulation results of general cross networks.}
\label{tableB}
\begin{tabular}{|c|c|c|c|c|}
  \hline
  & $\frac{cost(R^{NO})}{cost(R^*)}$ & $\frac{cost(R^{BIP/sw})}{cost(R^*)}$ & $\frac{cost(R^{BIP})}{cost(R^*)}$ & $\frac{cost(R^{D})}{cost(R^*)}$ \\
   \hline
  $N=13$ & 1.0668 $\pm$ 0.0293 & 1.1302 $\pm$ 0.0362 & 1.1747 $\pm$ 0.0401 & 1.2556 $\pm$ 0.0584 \\
  \hline
\end{tabular}
\end{table}

As Table \ref{tableA} shows, the energy consumption of the near-optimal assignment is close to that of the optimal assignment. Table \ref{tableA}  also shows that the proposed near-optimal assignment outperforms BIP and BIP with sweep rather considerably. This is in addition to the advantage of having a lower complexity. Based on Table \ref{tableA}, the distributed assignment performs the worst, but still provides a very low complexity alternative at the cost of about $40\%$ extra energy compared to the optimal assignment.
Increasing $N$ from $14$ to $18$ makes the energy gap between the optimal algorithm and the other sub-optimal algorithms shrink by a non-negligible amount. In particular, one should also note that the performance gap of the distributed algorithm  relative to the optimal solution shrinks faster than the other assignments by increasing the
size of the network.

For the general case, based on the results of Table \ref{tableB}, the difference between the optimal algorithm, on the one hand, and the near-optimal and distributed algorithms, on the other hand, is even less. Table \ref{tableB} shows that the proposed near-optimal assignment still outperforms both versions of BIP.

For larger networks, where the optimal assignment is too complex to find, we only present the results for the other assignments. In Fig. \ref{figCost_B}, we compare the total consumed energy of different assignments normalized with respect to the proposed near-optimal assignment. For each simulation point corresponding to a given number of nodes, $10,000$ random networks are generated, each having uniform distribution for the nodes. We run the algorithms on exactly the same networks and obtain the average of the total consumed energy over the $10,000$ networks.

\begin{figure}[h!]
  \centering
    \includegraphics[width=5in]{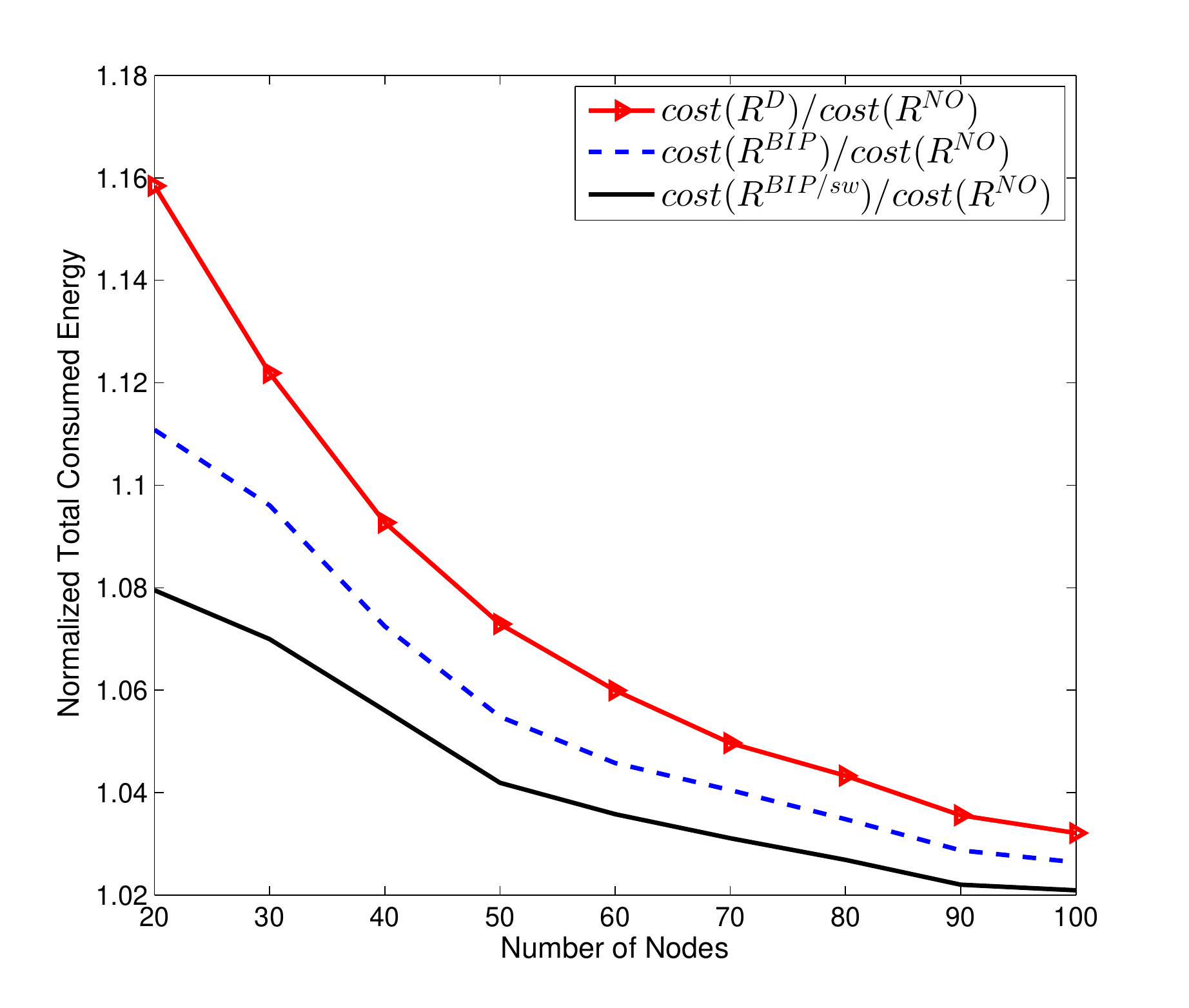}
     \caption{Energy consumption comparison of different range assignments.}
     \label{figCost_B}
\end{figure}

Fig. \ref{figCost_B} confirms the superiority of the proposed near-optimal assignment compared with BIP and BIP with sweep. It also shows that the distributed assignment performs close to the other assignments particularly for larger networks. In general, the gap between different algorithms shrinks as the size of the networks increases.

To compare the distributed algorithm with BIP (with and without sweeping procedure) for the general grid networks, we consider a two by two square grid network. We run the algorithms on exactly the same networks and obtain the average of the total consumed energy over $10,000$ networks. The comparison between the cost of the distributed algorithm and BIP (with and without sweeping procedure) is shown in Fig. \ref{figCost_G}. It can be seen that for larger number of nodes in the network, the difference between the algorithms is smaller. In particular the difference in energy consumption of the proposed distributed algorithm and BIP with sweep is less than $5\%$ for networks of size $N=40$ or larger.

\begin{figure}[h!]
  \centering
    \includegraphics[width=5in]{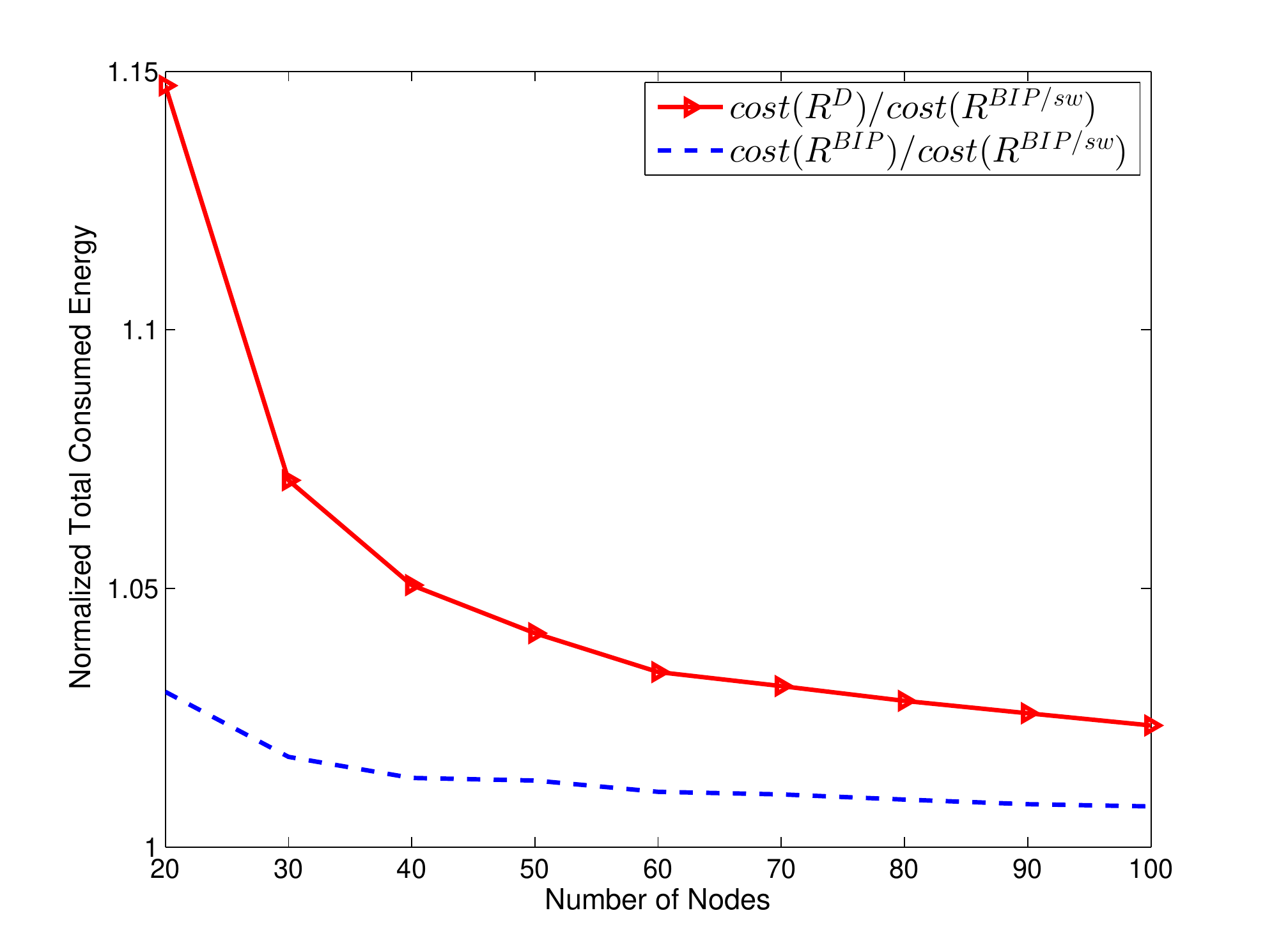}
     \caption{Comparison of the energy consumption of different range assignments for two by two square gird networks.}
     \label{figCost_G}
\end{figure}

\section{Conclusion}
\label{sec:con}
We studied a special 2-D wireless network, where $N$ nodes are located on two perpendicular lines. We proposed three range assignments for energy-efficient broadcasting from a source node anywhere in the network to all the other nodes. We first proposed an optimal assignment with polynomial complexity followed by a near-optimal assignment with linear complexity $\mathcal{O}(N)$, and a distributed assignment with complexity independent of $N$. We compared the proposed range assignments with BIP (with and without sweep), and observed that our near-optimal assignment outperforms both versions of BIP, while the distributed algorithm performs close to it. The advantage of the proposed algorithms over BIP algorithms is further seen when one notes the higher complexity $\mathcal{O}(N^2)$ of the BIP algorithms in comparison with linear and constant complexities of the proposed near-optimal and distributed algorithms, respectively.
We demonstrated that our proposed distributed algorithm can also be used for more general two-dimensional networks, where the nodes are located on a grid consisting of perpendicular line-segments. In such network configurations also, the proposed distributed algorithm performs close to BIP algorithms, particularly for larger networks.
To the best of our knowledge, this is the first study presenting an optimal solution for the minimum-energy broadcasting problem for a 2-D network.

\bibliographystyle{ieeetr} 

\bibliography{refs}  

\end{document}